\documentclass[11pt]{article}

\usepackage{titlesec}

\usepackage{amsfonts,amsmath,amssymb,amsthm,color}
\usepackage[margin=1.0in]{geometry}
\usepackage{microtype}
\definecolor{weborange}{rgb}{.8,.3,.3}
\definecolor{webblue}{rgb}{0,0,.8}
\definecolor{internallinkcolor}{rgb}{0,.5,0}
\definecolor{externallinkcolor}{rgb}{0,0,.5}

\providecommand{\remove}[1]{}

\usepackage{amsthm} % in order to remove warnings such as "destination with the same identifier (name{corollary.1.3}) has been already used, duplicate ignored" the order of packages should be hyperref, amsthm, cleveref

\usepackage{comment}

\usepackage{enumerate}
\usepackage[labelfont=bf]{caption}

\usepackage[%pagebackref,
hyperfootnotes=false,
colorlinks=true,
urlcolor=externallinkcolor,
linkcolor=internallinkcolor,
filecolor=externallinkcolor,
citecolor=internallinkcolor,
breaklinks=true,
pdfstartview=FitH,
pdfpagelayout=OneColumn]{hyperref}

\usepackage[backend=bibtex,style=ieee-alphabetic,natbib=true,backref=true]{biblatex} %added
\usepackage{cleveref}
\usepackage{xspace}
\usepackage{dsfont}

\usepackage{graphicx}

\providecommand{\remove}[1]{}

\ifdefined\IsDraft
\newcommand{\authnote}[2]{{\bf [{\color{red} #1's Note:} {\color{blue} #2}]}}
\else
\newcommand{\authnote}[2]{}
\fi

\remove{
	
	\titleclass{\subsubsubsection}{straight}[\subsection]
	
	\newcounter{subsubsubsection}[subsubsection]
	\renewcommand\thesubsubsubsection{\thesubsubsection.\arabic{subsubsubsection}}
	\titleformat{\subsubsubsection}
	{\normalfont\normalsize\bfseries}{\thesubsubsubsection}{1em}{}
	\titlespacing*{\subsubsubsection}
	{0pt}{3.25ex plus 1ex minus .2ex}{1.5ex plus .2ex}
	
	\makeatletter
	\renewcommand\paragraph{\@startsection{paragraph}{5}{\z@}%
		{3.25ex \@plus1ex \@minus.2ex}%
		{-1em}%
		{\normalfont\normalsize\bfseries}}
	\renewcommand\subparagraph{\@startsection{subparagraph}{6}{\parindent}%
		{3.25ex \@plus1ex \@minus .2ex}%
		{-1em}%
		{\normalfont\normalsize\bfseries}}
	\def\toclevel@subsubsubsection{4}
	\def\toclevel@paragraph{5}
	\def\toclevel@paragraph{6}
	\def\l@subsubsubsection{\@dottedtocline{4}{7em}{4em}}
	\def\l@paragraph{\@dottedtocline{5}{10em}{5em}}
	\def\l@subparagraph{\@dottedtocline{6}{14em}{6em}}
	\makeatother
	
	\setcounter{secnumdepth}{4}
	\setcounter{tocdepth}{4}
} % remove

\newenvironment{algorithm}{\begin{mybox} \vspace{-.1in}\begin{algo}}{ \vspace{-.1in} \end{algo}\end{mybox}}

\newenvironment{mybox}{\begin{center}\begin{tabular}{|p{0.97\linewidth}|c|}   \hline} {  \\ \hline \end{tabular} \end{center}}			

\let\originalleft\left
\let\originalright\right
\renewcommand{\left}{\mathopen{}\mathclose\bgroup\originalleft}
\renewcommand{\right}{\aftergroup\egroup\originalright}

	\newcommand{\eg}  {e.g.,\xspace}

	\newcommand{\wlg} {without loss of generality\xspace}

	\newcommand{\ceil}[1]{\left\lceil #1 \right\rceil}
	
	\newcommand{\signn}[1]{\sign\paren{#1}}
	\newcommand{\ip}[1]{\iprod{#1}}
	\newcommand{\iprod}[1]{\langle #1 \rangle}
	
	\newcommand{\set}[1]{\ens{#1}}

	\newcommand{\paren}[1]{\left(#1\right)}

	\newcommand{\floor}[1]{\left \lfloor#1 \right \rfloor}

	\newcommand{\norm}[1]{\left\lVert#1\right\rVert}

	\newcommand{\eqdef}{:=}

        \newcommand{\R}{{\mathbb R}}

	\newcommand{\zo}{\set{0,1}}
	\newcommand{\oo}{\set{-1,1}}

	\newcommand{\condition}{\;\ifnum\currentgrouptype=16 \middle\fi|\;}

	\newcommand{\eps}{\varepsilon}

	\newcommand{\la}{\gets}

	\newcommand{\Gen}{\mathsf{Gen}}
	\newcommand{\Trace}{\mathsf{Trace}}

	\newcommand{\Supp}{\operatorname{Supp}}

	\newcommand{\MathAlg}[1]{\mathsf{#1}}
	
	\newcommand{\sign}{\MathAlg{sign}}

	\newcommand{\Ensuremath}[1]{\ensuremath{#1}\xspace}

	\newcommand{\tth}[1]{\Ensuremath{#1^{\rm th}}}
	
	\newcommand{\ith}{\tth{i}}
	\newcommand{\jth}{\tth{j}}

	\renewcommand{\cref}{\Cref}
	\usepackage{aliascnt}
	
	\newtheorem{theorem}{Theorem}[section]
	
	\newaliascnt{lemma}{theorem}
	\newtheorem{lemma}[lemma]{Lemma}
	\aliascntresetthe{lemma}
	\crefname{lemma}{Lemma}{Lemmas}

	\newaliascnt{observation}{theorem}
	
	\aliascntresetthe{observation}
	\crefname{observation}{Observation}{Observation}

	\newaliascnt{claim}{theorem}
	\newtheorem{claim}[claim]{Claim}
	\aliascntresetthe{claim}
	\crefname{claim}{Claim}{Claims}
	
	\newaliascnt{corollary}{theorem}
	\newtheorem{corollary}[corollary]{Corollary}
	\aliascntresetthe{corollary}
	\crefname{corollary}{Corollary}{Corollaries}

	\newaliascnt{construction}{theorem}
	
	\aliascntresetthe{construction}
	\crefname{construction}{Construction}{Constructions}
	
	\newaliascnt{fact}{theorem}
	\newtheorem{fact}[fact]{Fact}
	\aliascntresetthe{fact}
	\crefname{fact}{Fact}{Facts}
	
	\newaliascnt{proposition}{theorem}
	
	\aliascntresetthe{proposition}
	\crefname{proposition}{Proposition}{Propositions}
	
	\newaliascnt{conjecture}{theorem}
	
	\aliascntresetthe{conjecture}
	\crefname{conjecture}{Conjecture}{Conjectures}

	\newaliascnt{definition}{theorem}
	\newtheorem{definition}[definition]{Definition}
	\aliascntresetthe{definition}
	\crefname{definition}{Definition}{Definitions}
	
	\newaliascnt{notation}{theorem}
	
	\aliascntresetthe{notation}
	\crefname{notation}{Notation}{Notation}
	
	\newaliascnt{assertion}{theorem}
	
	\aliascntresetthe{assertion}
	\crefname{assertion}{Assertion}{Assertion}
	
	\newaliascnt{assumption}{theorem}
	
	\aliascntresetthe{assumption}
	\crefname{assumption}{Assumption}{Assumption}

	\newaliascnt{remark}{theorem}
	\newtheorem{remark}[remark]{Remark}
	\aliascntresetthe{remark}
	\crefname{remark}{Remark}{Remarks}
	
	\newaliascnt{question}{theorem}
	
	\aliascntresetthe{question}
	\crefname{question}{Question}{Question}

	\newaliascnt{example}{theorem}
	\ifdefined\excludeexample
	\else
	
	\fi

	\aliascntresetthe{example}
	\crefname{exmaple}{Example}{Examples}

	\crefname{equation}{Equation}{Equations}

	\newaliascnt{proto}{theorem}
	
	\newtheorem{proto}[proto]{Protocol}
	
	\aliascntresetthe{proto}
	\crefname{proto}{protocol}{protocols}
	\newaliascnt{algo}{theorem}
	\newtheorem{algo}[algo]{Algorithm}
	\aliascntresetthe{algo}
	\crefname{algo}{algorithm}{algorithms}

	\newaliascnt{expr}{theorem}
	\newtheorem{expr}[expr]{Experiment}
	\aliascntresetthe{expr}
	\crefname{experiment}{experiment}{experiments}

	\newaliascnt{gm}{theorem}
	\newtheorem{gm}[gm]{Game}
	\aliascntresetthe{gm}
	\crefname{game}{game}{games}

	\newcommand{\stepref}[1]{Step~\ref{#1}}

	\def\FullBox{$\Box$}
	\def\qed{\ifmmode\qquad\FullBox\else{\unskip\nobreak\hfil
			\penalty50\hskip1em\null\nobreak\hfil\FullBox
			\parfillskip=0pt\finalhyphendemerits=0\endgraf}\fi}
	
	\def\qedsketch{\ifmmode\Box\else{\unskip\nobreak\hfil
			\penalty50\hskip1em\null\nobreak\hfil$\Box$
			\parfillskip=0pt\finalhyphendemerits=0\endgraf}\fi}

	{\proof[#1]\list{}{\leftmargin=1cm\rightmargin=1cm}}
	{\endproof%
		\vspace{10pt}
	}
	
	\newcommand{\eex}[2]{\Ex_{#1}\left[#2\right]}
	\newcommand{\ex}[1]{\Ex\left[#1\right]}
	\newcommand{\Ex}{{\mathrm E}}
	\renewcommand{\Pr}{{\mathrm {Pr}}}
	\newcommand{\pr}[1]{\Pr\left[#1\right]}
	\newcommand{\ppr}[2]{\Pr_{#1}\left[#2\right]}

	\newcommand{\Ac}{\mathsf{A}}
	
	\newcommand{\PAP}{\mathsf{PAP}}

	\newcommand{\Fc}{\mathsf{F}}
	\newcommand{\Gc}{\mathsf{G}}

	\newcommand{\Mc}{\mathsf{M}}
	\newcommand{\tMc}{\widetilde{\Mc}}

	\newcommand{\Hc}{\mathsf{H}}

	\newcommand{\ens}[1]{\{#1\}}
	\newcommand{\size}[1]{\left|#1\right|}

	\newcommand{\cW}{{\cal{W}}}

	\def\cB{{\cal B}}
	
	\def\cD{{\cal D}}

	\def\cJ{{\cal J}}

	\def\cP{{\cal P}}
	
	\def\cR{{\cal R}}
	\def\cS{{\cal S}}
	\def\cT{{\cal T}}
	
	\def\cV{{\cal V}}
	\def\cW{{\cal W}}
	\def\cX{{\cal X}}
	
	\def\cY{{\cal Y}}
	\def\cZ{{\cal Z}}
	\def\bbN{{\mathbb N}}
	
	\def\bbR{{\mathbb R}}

	\newcommand{\Tableofcontents}{
		\thispagestyle{empty}
		\pagenumbering{gobble}
		\clearpage
		\tableofcontents
		\thispagestyle{empty}
		\clearpage
		\pagenumbering{arabic}
	}

	\newcommand{\mypar}[1]{\vspace{4pt} \noindent {\sc \textbf{#1}}~~~}

	\newcommand{\indic}[1]{\mathds{1}_{\set{#1}}}

	\newcommand{\tx} {\tilde{x}}

	\newcommand{\tq} {\tilde{q}}

	\newcommand{\tX} {\tilde{X}}

	\newcommand{\bq}{\bold{q}}
	\newcommand{\bs}{\bold{s}}
	\newcommand{\by}{\bold{y}}
	\newcommand{\bY}{\bold{Y}}
	\newcommand{\bb}{\bold{b}}
	\newcommand{\bp}{\bold{p}}
	\newcommand{\bP}{\bold{P}}
	\newcommand{\bX}{\bold{X}}
	\newcommand{\ba}{\bold{a}}
	
	\newcommand{\bx}{\bold{x}}
	\newcommand{\br}{\bold{r}}

	\newcommand{\bv}{\bold{v}}

	\newcommand{\bz}{\bold{z}}

	\newcommand{\OPT}{{\rm OPT}}
	
	\newcommand{\COST}{{\rm COST}}

	\makeatletter
	\let\xx@thm\@thm
	\AtBeginDocument{\let\@thm\xx@thm}
	\makeatother

	\newcommand{\st}{{\sf st}}

        \renewcommand{\epsilon}{\varepsilon}

\newcommand{\ignore}[1]{}

\title{Smooth Lower Bounds for Differentially Private Algorithms\\ via Padding-and-Permuting Fingerprinting Codes}

\ifdefined\IsAnonymous
\author{Anonymized for Submission}
\else
\author{
Naty Peter\thanks{Department of Computer Schience, University of Toronto. Work conducted while at Georgetown University. E-mail: {\tt naty.peter@utoronto.ca}. Work supported in part by the Massive Data Institute at Georgetown University.} 
\and Eliad Tsfadia\thanks{Department of Computer Science, Georgetown University. E-mail: \texttt{eliadtsfadia@gmail.com}.  Work supported in part by the Fulbright Program and a gift to Georgetown University.}
\and Jonathan Ullman\thanks{Khoury College of Computer Sciences, Northeastern University.  E-mail: \texttt{jullman@ccs.neu.edu}. Work supported by NSF awards CNS-2120603, CNS-2232692, and CNS-2247484.}
}
\fi
\begin{document}

\maketitle

\begin{abstract}
	
Fingerprinting arguments, first introduced by \citet{BUV14}, are the most widely used method for establishing lower bounds on the sample complexity or error of approximately differentially private (DP) algorithms.  Still, there are many problems in differential privacy for which we don't know suitable lower bounds, and even for problems that we do, the lower bounds are not smooth, and become vacuous when the error is larger than some threshold.

In this work, we present a new framework and tools to generate smooth lower bounds that establish strong lower bounds on the sample complexity of differentially private algorithms satisfying very weak accuracy.  We illustrate the applicability of our method by providing new lower bounds in various settings: 
\begin{enumerate}
    \item A tight lower bound for DP averaging in the low-accuracy regime, which in particular implies a lower bound for the private $1$-cluster problem introduced by \citet{NSV16}.
    \item A lower bound on the additive error of DP algorithms for approximate $k$-means clustering and general $(k,z)$-clustering, as a function of the multiplicative error, which is tight for a constant multiplication error.
    \item A lower bound for estimating the top singular vector of a matrix under DP in low-accuracy regimes, which is a special case of the DP subspace estimation problem studied by \citet{SS21}.
\end{enumerate}

Our new tools are based on applying a padding-and-permuting transformation to a fingerprinting code. However, rather than proving our results using a black-box access to an existing fingerprinting code (e.g., Tardos' code \citep{Tardos08}), we develop a new fingerprinting lemma that is stronger than those of \citet{DworkSSUV15,BunSU17}, and prove our lower bounds directly from the lemma.   Our lemma, in particular, gives a simpler fingerprinting code construction with optimal rate (up to polylogarithmic factors) that is of independent interest.  
\end{abstract}

\Tableofcontents

\section{Introduction}

Differentially private (DP) \citep{DMNS06} algorithms provide a strong guarantee of privacy to the individuals who contribute their data.  Informally, a DP algorithm takes data from many individuals and guarantees that no attacker, regardless of their knowledge or capabilities, can learn much more about any one individual than they would have if that individual's data had never been collected.  There is a large body of literature on DP algorithms, and DP algorithms have now been deployed by large technology companies and government organizations.

\mypar{DP Averaging.}
As a running example, 
suppose our input dataset is $x_1,\dots,x_n \in \R^d$ and our goal is to estimate their average $\frac1n \sum_i x_i$.  Since DP requires us to hide the influence of one data point on the average, we intuitively need to assume some kind of bounds on the data.  A common way to bound the data is to assume that it lies in some ball of radius $r$, so there exists a center $c \in \R^d$ such that $\|x_i - c\|_2 \leq r$.  Our goal is then to output a DP average $\hat{x}$ such that, with high probability,
\begin{equation*}\label{eq:intro1}
\left\| \hat{x} - \frac1n \sum_i x_i \right\|_2 \leq \lambda r
\end{equation*}
If we assume that the location of the center of the ball is \emph{known}, then the natural DP algorithm is to clip the data to lie in this known ball, and perturb the true average with noise from a Laplacian or Gaussian distribution of suitable variance.  One can show that the average will satisfy the error guarantee above if the dataset has at least $n \gtrsim \sqrt{d}/\lambda$ samples.  

However, in many applications of DP involving real data, we do not want to assume that the center of the ball is known.  For example, algorithms like clustering, covariance estimation, and PCA are often applied to datasets as a preprocessing step to understand the general properties, and we cannot assume the user already knows the location of the data.  Thus, we want to assume that the data lies in a ball whose center is \emph{unknown}.  For this problem, the FriendlyCore algorithm of \citet{FriendlyCore22} is able to achieve the same error guarantee even when the location of the data is \emph{a priori} unknown, provided $n \gtrsim \sqrt{d}/\lambda$.

\mypar{Lower Bounds for DP Averaging.} The work of \citet{BUV14} proved that $\Omega(\sqrt{d}/\lambda)$ samples are required for DP averaging, when $\lambda \lesssim 1$. Their work introduced the method of \emph{fingerprinting codes} to differential privacy, and this technique has become the standard approach for proving lower bounds for differentially private algorithms, either by reduction to the averaging problem or by non-black-box use of the fingerprinting technique (see Related Work).

The lower bound of \citet{BUV14} has a significant drawback that it only applies when $\lambda \lesssim 1$, and the lower bound on the sample complexity is vacuous for $\lambda \geq 1$.  This limitation is inherent to the way these lower bounds work, since they construct a hard distribution over the hypercube $\{-1,1\}^d$, which lies in a \emph{known} ball of radius $r = \sqrt{d}$.  So the DP algorithm that outputs $\hat{x} = \vec{0}$ satisfies
$$
\left\| \hat{x} - \frac1n \sum_i x_i \right\|_2 \leq \sqrt{d} = r.
$$
Thus there is no need for any samples when the error parameter is $\lambda \geq 1$, so the lower bound fully captures the hardness of DP averaging when the location of the data is known.  

However, when the location of the ball is unknown, even finding a low-accuracy DP average with $\lambda \geq 1$ is non-trivial.
In this work, we develop general tools for generating hard-instances for such types of problems. In particular, for DP-averaging, our results implies that $n = \tilde{\Omega}(\sqrt{d}/\lambda)$ samples are required for all $\lambda$, yielding that the above algorithms are essentially optimal.
While this low-accuracy regime may seem like an intellectual curiosity, it turns out that low-accuracy approximations of this sort are quite useful for a variety of DP approximation algorithms, and we show that our technique implies new lower bounds for other widely studied problems---computing a DP $k$-means clustering with a constant multiplicative approximation, and finding a DP top singular vector---that crucially rely on the fact that our lower bound applies to the low-accuracy regime.

\subsection{Our Results}

\subsubsection{Main Hardness Results}\label{sec:intro:main-results}
Our core new technique is a variation of \cite{BUV14}'s method for creating strong error robust fingerprinting codes.
That is, a \emph{padding-and-permuting} transformation applied to a (weak error robust) fingerprinting codes. In this work, we use padding size that varies as a function of the accuracy guarantee, in contrast to \cite{BUV14} that use fixed-size padding (which suffices for robustness). Using such instances with large padding allows us to smoothly shrink the radius of the points in the hard instances while preserving the hardness.  This technique allows us to give the following general construction of a hard problem in DP.  Our results will follow from applying this hardness in a black-box way.

\begin{definition}[$b$-Marked Column]
    Given a matrix $X = (x_i^j)_{i \in [n], j \in [d]}  \in \oo^{n \times d}$ and $b \in \oo$, we say that a column $j \in [d]$ is $b$-marked if $x^j_1 = x^j_2= \ldots = x^j_n = b$. 
    We denote by $\cJ_X^b \subseteq [d]$ the set of $b$-marked columns of $X$.
\end{definition}

\begin{definition}[Strongly Agrees]\label{def:intro:strongly-agree}
	We say that a vector $q = (q^1,\ldots,q^d)$ strongly-agrees with a matrix $X \in \oo^{n \times d}$, if 
	\begin{align*}
		\forall b\in \oo: \quad \size{\set{j \in \cJ_X^b \colon q^j = b}} \geq 0.9 \size{\cJ_X^b}.
	\end{align*}
	(i.e., for both $b \in \oo$, $q$ agrees with at least $90\%$ of the $b$-marked columns of $X$).
\end{definition}

\begin{definition}[$(\alpha,\beta)$-Weakly-Accurate Mechanism]\label{def:intro:weakly-accurate}
	Let $\alpha, \beta \in (0,1]$. 
	We say that a mechanism $\Mc \colon \oo^{n \times d} \rightarrow [-1,1]^d$ is $(\alpha,\beta)$-weakly-accurate if for every input $X = (x_1,\ldots,x_n) \in  \oo^{n \times d}$ with $\size{\cJ_X^1}, \size{\cJ_X^{-1}} \geq \frac12(1-\alpha)d$,
	the probability that $\Mc(X)$ strongly-agrees with $X$ is at least $\beta$.
\end{definition}

Namely, for a small $\alpha$, the only requirement from an $(\alpha,\beta)$-weakly-accurate mechanism is to agree (w.p. $\beta$) with most of the $1$-marked and $(-1)$-marked columns, but only when almost half of the input columns are $1$-marked, and almost all the other half of the columns are $(-1)$-marked (otherwise, there is no restriction on the output).

The following theorem captures our general tool for lower bounding DP algorithms.

\begin{theorem}\label{thm:intro:base_tool}
    If $\Mc \colon (\oo^d)^n \rightarrow [-1,1]^d$ is an $(\alpha,\beta)$-weakly-accurate $\paren{1,\frac{\beta}{4n}}$-DP mechanism, then $n \geq  \Omega(\sqrt{\alpha d} / \log^{1.5} (\alpha d/\beta))$.
\end{theorem}

So in order to prove a lower bound for a specific task, it suffices to prove that the assumed utility guarantee implies \cref{def:intro:weakly-accurate}.

Theorem \ref{thm:intro:base_tool} can be proven by combining our padding-and-permuting technique with an optimal \emph{fingerprinting code}---such as Tardos' code \citep{Tardos08}---in a black-box way.  Moreover, in contrast with most recent constructions of DP lower bounds that use only a so-called \emph{fingerprinting lemma}, our techniques seem to require the using of a fingerprinting code.
Specifically, although fingerprinting lemmas are simpler and more flexible, they require a stronger notion of accuracy that does not fit our padding-and-permuting construction, whereas fingerprinting codes require only an extremely weak notion of accuracy to obtain hardness.  In an effort to unify and simplify the techniques used to prove DP lower bounds, we give an alternative proof that makes use of a \emph{new fingerprinting lemma} that only requires very weak accuracy to obtain hardness (see \cref{sec:ProofOverview} for more details).

We also consider an extension of \cref{def:intro:weakly-accurate} to cases where the mechanism receive a dataset which consists of $k$ clusters and is required to output a point that strongly-agrees with one of the clusters.

\begin{definition}[$(k,\alpha,\beta)$-Weakly-Accurate Mechanism]\label{def:intro:k-weakly-accurate}
	Let $\alpha,\beta \in (0,1]$ and $n,k,d \in \bbN$ such that $n$ is a multiple of $k$.
	We say that a mechanism $\Mc \colon \oo^{n \times d} \rightarrow [-1,1]^{d}$ is \emph{$(k,\alpha,\beta)$-weakly-accurate} if the following holds: Let $X = (x_1,\ldots,x_n) \in \oo^{n \times d}$ be an input such that for every $t \in [k]$ and every $b \in \oo$ it holds that $\size{\cJ^b_{X_t}} \geq \frac12(1-\alpha)d$ for $X_t = (x_{(t-1) n/k + 1}, \: \ldots, \: x_{tn/k}) \in (\oo^d)^{n/k}$. Then
	\begin{align*}
		\pr{\exists t \in [k]\text{ s.t. }\Mc(X)\text{ strongly-agrees with } X_t} \geq \beta.
	\end{align*}
\end{definition}
Note that $(k=1,\alpha,\beta)$-weakly-accurate is equivalent to $(\alpha,\beta)$-weakly-accurate.

Using $k$ independent padding-and-permuting fingerprinting code instances, we prove the following theorem.

\begin{theorem}[Extension of \cref{thm:intro:base_tool}]\label{thm:intro:extended_tool}
   Let $\alpha,\beta \in (0,1]$ and $n,k,d \in \bbN$ such that $n$ is a multiple of $k$. If $\Mc \colon (\oo^d)^n \rightarrow [-1,1]^d$ is an $(k,\alpha, \beta)$-weakly-accurate $\paren{1,\frac{\beta}{4n}}$-DP mechanism, then $n \geq \Omega(k \sqrt{\alpha d} / \log^{1.5} (\alpha d/\beta))$.
\end{theorem}

We prove \cref{thm:intro:base_tool,thm:intro:extended_tool} using a more general framework that we developed in this work that might be useful for other types of problems with more complicated output spaces (e.g., subspace or covariance estimation). We refer to \cref{sec:technique:framework} for more details.

\subsubsection{Application: Averaging and $1$-Cluster}

We first formally state our tight lower bound for DP averaging in the low-accuracy regimes. 
We start by defining a $(\lambda, \beta)$-estimator for averaging.

\begin{definition}[$(\lambda, \beta)$-Estimator for Averaging]\label{def:intro:avg-est}
    A mechanism $\Mc \colon \bbR^+ \times (\bbR^{d})^n \rightarrow \bbR^d$ is \emph{$(\lambda, \beta)$-estimator for averaging} if given $\gamma \geq 0$ and $x_1,\ldots,x_n \in \oo^d$ with $\max_{i,j \in [n]} \norm{x_i - x_j}_2  \leq \gamma$, it holds that 
	\begin{align*}
		\pr{\norm{\Mc(\gamma,x_1,\ldots,x_n) - \frac1n \sum_{i=1}^n x_i}_2 \leq \lambda \gamma} \geq \beta.
	\end{align*}
\end{definition}

It is well-known how to construct DP $(\lambda, \beta=0.99)$-estimators using $\tilde{O}(\sqrt{d}/\lambda)$ points.

\begin{fact}[Known upper bounds (\eg \cite{KV18,FriendlyCore22,AL22,NarayananME22})]
    For  $n = \tilde{O}(\sqrt{d}/\lambda)$, there exists an $(1,\frac1{n^2})$-DP $(\lambda, \beta=0.99)$-estimator for averaging.
\end{fact}

Using \cref{thm:intro:base_tool}, we prove a matching lower bound (up to low-order terms).

\begin{theorem}[Our averaging lower bound]\label{thm:intro:lower_bound:avg}
    If $\Mc \colon \bbR^+ \times (\bbR^{d})^n \rightarrow \bbR^d$ is an \emph{$(\lambda, \beta)$-estimator for averaging} for $\lambda \geq 1$ and $\Mc(\gamma,\cdot)$ is $\paren{1,\frac{\beta}{4n}}$-DP for every $\gamma \geq 0$, then $n \geq \Omega\paren{\frac{\sqrt{d}/\lambda}{\log^{1.5} \paren{\frac{d}{\beta \lambda}}}}$.
\end{theorem}

\paragraph{Immediate Application: The $1$-Cluster Problem.}

An interesting application of \cref{thm:intro:lower_bound:avg} is\remove{the first meaningful} a simple lower bound for the $1$-cluster problem that is widely used in DP clustering algorithms.

In the $1$-cluster problem, we are given $n$ points from a finite domain $\cX^d$ for $\cX \subseteq \bbR$, and a parameter $t \leq n$. The goal is to identify a $d$-dimensional ball that contains almost $t$ point, such that the size of the ball is not too far from the optimum. Formally,

\begin{definition}[$(\lambda, \: \beta, \: t_{low}, \: s)$-Estimator for $1$-Cluster, \cite{NSV16}]
    A mechanism $\Mc \colon (\cX^d)^n \times [n] \rightarrow \bbR^+ \times \bbR^d$ is an \emph{$(\lambda, \beta, t_{low}, s)$-estimator for $1$-cluster} if given $\cS \in (\cX^d)^n$ and $t \in [t_{low}, n]$ as inputs, it outputs $r \geq 0$ and $c \in \bbR^d$ such that the following holds with probability at least $\beta$:
    \begin{enumerate}
        \item The ball of radius $r$ around $c$ contains at least $t - s$ points from $\cS$, and

        \item Let $r_{opt}$ be the radius of the smallest ball in $\cX^d$ containing at least $t$ input points. Then $r \leq \lambda \cdot r_{opt}$.
    \end{enumerate}
\end{definition}

It is well-known how to privately solve this problem with a constant $\lambda$, whenever $t_{low} \geq \tilde{\Theta}(\sqrt{d})$.

\begin{fact}[Upper bounds \cite{NSV16,NS18_1Cluster}, simplified]\label{fact:intro:1-cluster}
    There exists an $(1,\frac1{n^2})$-DP, $(\lambda = \Theta(1), \: \beta = 0.99, \: t_{low} = \tilde{\Theta}(\sqrt{d}), \: s = \tilde{\Theta}(1))$-estimator for $1$-cluster.
\end{fact}

Therefore, we conclude by \cref{thm:intro:lower_bound:avg} the following tight lower bound (up to low order terms) which is essentially an immediate corollary of our averaging lower bound.

\begin{corollary}[Our $1$-cluster lower bound]
    If $\Mc$ is $\paren{1,\frac{\beta}{4n}}$-DP and $(\lambda, \: \beta, \: t_{low}, \:s=n-1)$-Estimator for $1$-Cluster for $\lambda \geq 1$, then $t_{low} \geq \Omega\paren{\frac{\sqrt{d}/\lambda}{\log^{1.5} \paren{\frac{d}{\beta \lambda}}}}$.
\end{corollary}

We remark that for these specific tasks of DP-averaging/1-cluster, a recent work of \cite{NarayananME22}, which provides a lower bound for user-level DP averaging, can also be used to prove a similar statement to \cref{thm:intro:lower_bound:avg} (up to poly-logarithmic terms).
Their technique is very different from ours, and in particular, do not extend to proving \cref{thm:intro:base_tool,thm:intro:extended_tool} which serve as simple tools for proving lower bounds for the other problems. We refer to \cref{sec:technique:comparison} for a more detailed comparison.

\subsubsection{Application: Clustering}\label{sec:intro:clustering}

In $k$-means clustering, we are given a database $\cS$ of $n$ points in $\bbR^d$, and the goal is to output $k$ centers $C = (c_1,\ldots,c_k) \in (\bbR^d)^k$ that minimize
\begin{align*}
    \COST(C; \cS) \eqdef \sum_{x \in \cS} \min_{i \in [k]}\norm{x - c_i}_2^2.
\end{align*}

Similarly to prior works, we focus, \wlg, on input and output points in the $d$-dimensional unit ball $\cB_d \eqdef \set{x \in \bbR^d \colon \norm{x}_2 \leq 1}$.
The approximation quality of a DP algorithm is measured in the literature by two parameters: a multiplicative error $\lambda$, and an additive error $\xi$, defined below: 

\begin{definition}[$(\lambda,\xi, \beta)$-Approximation Algorithm for $k$-Means]\label{def:intro:k_means_approx}
    $\Mc\colon (\cB_d)^n \rightarrow (\cB_d)^k$ is an \emph{$(\lambda,\xi, \beta)$-approximation algorithm for $k$-means}, if for every $\cS \in (\cB_d)^n$ it holds that
    \begin{align*}
        \ppr{C \sim \Mc(\cS)}{\COST(C;\cS) \leq \lambda \cdot \OPT_k(\cS) + \xi} \geq \beta,
    \end{align*}
    where $\OPT_k(\cS) \eqdef \min_{C \in (\cB_d)^k} \COST(C;\cS)$.
\end{definition}

While non-private algorithms usually do not have an additive error, under DP an additive error is necessary. As far as we are aware, the only known lower bound on the additive error is the one of \cite{GuptaLMRT10} (stated for $k$-medians), which has been extended later by \cite{NguyenCX21} (Theorem 1.2). This lower bound essentially says that the additive error $\xi$ of any $(1,\frac1{n^2})$-DP algorithm for $k$-means (regardless of its multiplicative error) must be at least $\tilde{\Omega}(k)$. 
However, as far as we are aware, all known DP upper bounds have an additive error of at least $\Omega(k \sqrt{d})$.
In particular, this is also the situation in the state-of-the-art upper bounds that have constant multiplicative error.

\begin{fact}[General upper bounds \cite{KaplanSt18,Ghazi0M20,NguyenCX21}, simplified]
    There exists an $(1,\frac1{n^2})$-DP $(\Theta(1), \: \tilde{\Theta}(k \sqrt{d}))$-approximation algorithm for $k$-means.
\end{fact}

Furthermore, an additive error of $\tilde{\Theta}(k \sqrt{d})$ also appears in algorithms that provide utility only for datasets that are well-separated into $k$-clusters:  For a small parameter $\phi \in [0,1]$, a dataset $\cS \in (\cB^d)^n$ is called $\phi$-separated for $k$-means if $\OPT_k(\cS) \leq \phi^2\cdot \OPT_{k-1}(\cS)$ (\cite{OstrovskyRSS12}).

\begin{fact}[Upper bounds for well-separated instances \cite{ShechnerSS20,CKMST:ICML2021}, simplified]
	There exists an $(1,\frac1{n^2})$-DP  $\:(1 + O(\phi^2), \: \tilde{\Theta}(k \sqrt{d}))$-approximation algorithm for $k$-means of $\phi$-separated instances. 
\end{fact}

Using \cref{thm:intro:extended_tool} we provide the first tight lower bound on the additive error for algorithms with constant multiplication error.
Furthermore, since \cref{thm:intro:extended_tool} is proven using $k$ independent padding-and-permuting FPC instances which induces $k$ obvious clusters that are far from each other, our lower bound also matches the upper bounds for well-separated instances.

\begin{theorem}[Our $k$-means lower bound]\label{thm:intro:kmeans}
	Let $n,k,d \in \bbN$, $\lambda \geq 1$, $\beta \in (0,1]$ and $\xi \geq 0$ such that $n \geq k + 80\xi$.
    If $\Mc\colon (\cB_d)^n \rightarrow (\cB_d)^k$ is an $\paren{1,\frac{\beta}{4nk}}$-DP $(\lambda, \xi, \beta)$-approximation algorithm for $k$-means, then either $k \geq 2^{\Omega(d/\lambda)} \beta \lambda/d \:$ or $\: \xi \geq \Omega\paren{\frac{k \sqrt{d/\lambda}}{\log^{1.5}\paren{\frac{kd}{\beta \lambda}}}}$.
\end{theorem}

Note that \cref{thm:intro:kmeans} even suggests how we might expect the additive error to decrease, if we increase the multiplicative error (for such cases, however, there are no matching upper bounds). 

We remark that our proof is not tailored to $k$-means clustering. In \cref{sec:k-means} we state and prove an extension of Theorem~\ref{thm:intro:kmeans} to \emph{$(k,z)$-clustering} in which the cost is measured by the sum of the $\tth{z}$ powers of the distances ($k$-means clustering is the special case of $z=2$).

\subsubsection{Application: Top Singular Vector}

In this problem, we are given $n$ points $x_1,\ldots,x_n \in \cS_d \eqdef \set{v \in \bbR^d \colon \norm{v}_2 = 1}$ of unit norm as input, and the goal is to estimate the top (right) singular vector of the $n \times d$ matrix $X = (x_i^j)_{i \in [n], j \in [d]}$, which is the unit vector $v \in \cS_d$ that maximizes $\norm{X\cdot v}_2$.

The singular value decomposition of $X$ is defined by $X  = U \Sigma V^T$, where $U \in \bbR^{n \times n}$ and  $V \in \bbR^{d \times d}$ are unitary matrices. The matrix $\Sigma \in \bbR^{n \times d}$ is a diagonal matrix with non-negative entries $\sigma_1 \geq \sigma_2 \geq \ldots \geq \sigma_{\min\set{n,d}} \geq 0$ along the diagonal, called the singular values of $X$. The first column of $V$ is the top right singular vector.

\citet{DTTZ14} proved a general lower bound of $n = \tilde{\Omega}(\sqrt{d})$ for any algorithm that identifies a useful approximation to the top singular vector. Yet, \cite{SinghalS21} bypassed this lower bound under a distributional assumption which implies that the points are close to lying in a $1$-dimensional subspace, defined by having a small ratio $\sigma_{2}/\sigma_{1}$. They showed that when the points are Gaussian, and the above ratio is small, the sample complexity can be made independent of $d$.
They also consider the more general problem of estimating the span of the top $k$ singular vectors, which we do not consider in this work.

\begin{definition}[$(\lambda,\beta)$-Estimator of Top Singular Vector]\label{def:intro:top_sing_est}
    We say that $\Mc \colon [0,1] \times (\cS_d)^n \rightarrow \cS_d$ is an $(\lambda, \beta)$-estimator of top singular vector, if given an $n \times d$ matrix $X = (x_1,\ldots,x_n) \in (\cS_d)^{n}$ and an upper bound $\gamma \in [0,1]$ on $\sigma_2/\sigma_1$ as inputs, outputs a column vector $y \in \cS_d$ such that 
    \begin{align*}
        \ppr{y \sim \Mc(\gamma, X)}{\norm{X\cdot y}_2^2 \geq \norm{X\cdot v}_2^2 - \lambda \gamma n} \geq \beta,
    \end{align*}
    where $v$ denotes the top singular vector of $X$.
\end{definition}

Our next result is a lower bound on the sample complexity of singular vector estimation that is smooth with respect to the spectral gap parameter $\gamma$.

\begin{theorem}[Our lower bound]\label{thm:intro:top_sing}
    If $\Mc \colon [0,1] \times (\cS_d)^n \rightarrow \cS_d$ is an \emph{$(\lambda, \beta)$-estimator of top singular vector} for $\lambda \geq 1$ and $\Mc(\gamma,\cdot)$ is $\paren{1,\frac{\beta}{4n}}$-DP for every $\gamma \in [0,1]$, then $n \geq \Omega\paren{\frac{\sqrt{d}/\lambda}{\log^{1.5} \paren{\frac{d}{\beta \lambda}}}}$.
\end{theorem}

For comparison, \citet{DTTZ14} proved a lower bound of $\tilde{\Omega}(\sqrt{d})$ on the sample complexity but it only applies when $\gamma$ is larger than some specific constant, whereas our lower bound holds for the entire range of $\gamma$.

\subsection{Related Work}
The connection between fingerprinting codes and differential privacy was first introduced by Bun, Ullman, and Vadhan~\cite{BUV14}.  Subsequent work significantly simplified and generalized the method in a number of ways \cite{SteinkeU16,DworkSSUV15,BunSU17,SteinkeU17,KLSU19,CaiWZ21,KamathMS22,CaiWZ21,NarayananME22}, including the removal of fingerprinting codes and distilling the main technical component into a \emph{fingerprinting lemma}.  The price of this simplicity and generality is that the lower bounds rely on having a stronger accuracy requirement that constrains both marked and unmarked columns, whereas the hard instances constructed directly from fingerprinting codes require only a very weak accuracy requirement that constrains only marked columns.

Fingerprinting lower bounds have applied in many settings, including mean estimation \cite{BUV14,DworkSSUV15,KLSU19,NarayananME22}, adaptive data analysis \cite{HU14,SU15}, empirical risk minimization \cite{BassilyST14}, spectral estimation \cite{DTTZ14}, combining public and private data \cite{BassilyCMNUW20}, regression \cite{CaiWZ21,CaiWZ23}, sampling \cite{RaskhodnikovaSSS21}, Gaussian covariance estimation \cite{KamathMS22,Narayanan2024}, continual observation \cite{JainRSS23}, and unbiased private estimation \cite{KamathMRSSS23}. Of these works, all of them except for \cite{RaskhodnikovaSSS21} use fingerprinting lemmas, rather than the stronger construction of fingerprinting codes.  Except \cite{NarayananME22} (discussed in \cref{sec:technique:comparison}), all these lower bounds are based on the same type of hard instances that lead to the data being contained in a known ball of a given radius, and thus none of them have the smoothness property we desire.  Thus, we believe that our method could find other applications beyond the ones that are described in this paper.

\subsection{Paper Organization}
In \cref{sec:ProofOverview} we present a proof overview of \cref{thm:intro:base_tool} for the case $\beta \approx 1$ that uses an optimal fingerprinting code as black-box,  explain how we give a direct proof using a strong fingerprinting lemma that we developed in this work, and describe additional properties of our results.

 Notations, definitions and general statements used throughout the paper are given in \cref{sec:prelim}. Our strong fingerprinting lemma is stated and proved in \cref{sec:FPL}. In \cref{sec:framework} we present a general framework for proving DP lower bounds that is based on our fingerprinting lemma. In \cref{sec:FPT:PAP} we present our padding-and-permuting transformation, and prove \cref{thm:intro:base_tool,thm:intro:extended_tool} using our framework. In \cref{sec:applications} we prove \cref{thm:intro:lower_bound:avg,thm:intro:kmeans,thm:intro:top_sing}, which give our main applications. In \cref{sec:appendix:FPC} we show how to construct a simple fingerprinting code using our strong fingerprinting lemma.

\section{Our Technique}\label{sec:ProofOverview}

In this section, we present a proof overview of \cref{thm:intro:base_tool} for $\beta \approx 1$. In \cref{sec:ProofViaTardos} we present a simple variant of the proof that uses an optimal fingerprinting code as black-box (e.g., \citet{Tardos08}). In \cref{sec:ProofViaLemma} we explain how we actually avoid the use of Tardos' fingerprinting code by developing a strong fingerprinting lemma. In \cref{sec:technique:framework} we describe our more general framework, and in \cref{sec:technique:comparison} we make a detailed comparison with \cite{NarayananME22}.

\subsection{Proof via an Optimal Fingerprinting Code}\label{sec:ProofViaTardos}

\paragraph{Fingerprinting Code (FPC).}
An FPC consists of two algorithms: $\Gen$ and $\Trace$. Algorithm $\Gen$ on input $n$ outputs a codebook (matrix) $(x_i^j)_{i \in [n], j \in [d]} \in \oo^{n \times d}$ for $d = d(n)$, and a secret state $\st$. An adversary who controls a coalition $\cS \subseteq [n]$ only gets the rows $(x_i)_{i \in \cS}$ and is required to output $q=(q^1,\ldots,q^d) \in \oo^d$ that agrees with the ``marked" columns, i.e., columns $j \in [d]$ where %all the 
$x_i^j = b$ (for the same $b \in \oo$) for every $i \in \cS$. On ``unmarked" columns $j \in [d]$, there is no restriction and the adversary is allowed to choose $q^j$ arbitrarily.
Algorithm $\Trace$, given such legal $q$ (and the secret state $\st$), guarantees to output $i \in \cS$ with high probability (i.e., to reveal at least one of the coalition members).

Fingerprinting codes were originally introduced by \citet{BonehS98}.
\citet{Tardos08} constructed an optimal FPC of length $d_0 = \tilde{\Theta}(n^2)$, and \citet{BUV14} proved that Tardos' code is actually \emph{robust}, i.e., it enables tracing even when the adversary is allowed to be inconsistent with a small fraction of marked columns (say, $20\%$).\footnote{\cite{BUV14} did not try to optimize the constant in the fraction of errors, and only proved it for $4\%$. For the purpose of this proof sketch, we assume that the code is robust for $20\%$ errors. A formal proof that relies on their result must change the constant $0.9$ in \cref{def:intro:strongly-agree} to a constant larger than $0.96$.}

\paragraph{Padding-and-Permuting FPC.}

Given a robust FPC $(\Gen,\Trace)$ of length $d_0 = d_0(n)$, consider the padding-and-permuting (PAP) variant of it as the following pair of algorithms $(\Gen',\Trace')$:

\begin{algorithm}[$\Gen'$]
    \item Input parameters: Number of users $n \in \bbN$ and accuracy parameter $\alpha \in [0,1]$. Let $d_0 = d_0(n)$ be the codewords' length of $\Gen(n)$ and  let $d = d_0 + 2\ell$ for $\ell = \ceil{\frac{d_0}{2\alpha}}$.

    \item Operation:~
    \begin{enumerate}
        \item Sample a codebook $X \in \oo^{n \times d_0}$ along with a secret state $\st$ according to $\Gen(n)$.
        \item Append $\ell$ $1$-marked and $\ell$ $(-1)$-marked columns to the matrix $X$.
        \item Permute the columns of $X$ according to a random permutation $\pi$ over $[d]$.
        \item Output the resulting matrix $X' \in \oo^{n \times d}$ along with the new state $\st' = (\pi,\st)$.
    \end{enumerate}
\end{algorithm}

\begin{algorithm}[$\Trace'$]
    \item Input parameters: A weakly-accurate result $q = (q^1,\ldots,q^{d}) \in \oo^{d}$ and a secret state $\st' = (\pi,\st)$.

    \item Operation:~
    \begin{itemize}
        \item Output $\Trace(\tq,\st)$ for $\tq = (q^{\pi(1)}, \ldots, q^{\pi(d_0)}) \in \oo^{d_0}$.
    \end{itemize}
\end{algorithm}

Note that we set the padding length as a function of the accuracy parameter $\alpha$ such that  a weaker accuracy (i.e., smaller $\alpha$) results with a larger padding. This is crucial for creating hard instances in the regime where an $\alpha$-weakly-accurate mechanism must be accurate.

\paragraph{Proving \cref{thm:intro:base_tool}}

Let $(\Gen,\Trace)$ be a robust FPC with codewords' length $d_0 = \tilde{\Theta}(n^2)$ (e.g., \cite{Tardos08}).
Suppose that we sample $X' = (x_1',\ldots,x_n') \in \oo^{n \times d}$ according to $\Gen'(n, \alpha)$. By construction, $X'$ contains at least $\frac12(1-\alpha)d$ $b$-marked columns, for both $b \in \oo$. Therefore, if $\Mc \colon \oo^{n \times d} \rightarrow \oo^{d}$ is $\alpha$-weakly-accurate, then the output $q \in \oo^d$ of $\Mc(X')$ must agree with $90\%$ of the marked columns of $X'$. But because the columns are randomly permuted, $\Mc$ cannot distinguish between marked columns from the padding and marked columns from the original codebook $X$. This means that it must agree with a similar fraction of marked columns of the codebook $X$, which enables tracing since the code is robust. Therefore, we conclude that such a weakly-accurate mechanism cannot be DP unless $n \geq \tilde{\Omega}(\sqrt{d_0})$, i.e., $n \geq \tilde{\Omega}(\sqrt{\alpha d})$.

We remark that handling smaller values of $\beta$ (the success probability of $\Mc$) creates more technical challenges that we ignore for the purpose of this overview.

\subsection{Proof via a Strong Fingerprinting Lemma}\label{sec:ProofViaLemma}

The disadvantage of using FPC as black-box for DP lower bounds is the fact that \cite{Tardos08}'s analysis is quite involved, so it is hard to gain an end-to-end understanding of the process, and in particular, to generate hard instances in more complicated settings (e.g., exponential families \cite{KamathMS22}).
Therefore, later results simplified the construction and especially the analysis, at the cost of considering more restricted adversaries that must estimate the average of most coordinates, and not just the ``marked" ones (which usually suffices for the aggregation tasks we are interested in). This led to the development of the Fingerprinting Lemma (described below) which serves as the most common tool for proving lower bounds for approximate DP algorithms. 

\begin{lemma}[Original Fingerprinting Lemma \cite{BunSU17,DworkSSUV15}]\label{lem:intro:origFPL}
    Let $f \colon \oo^n \rightarrow [-1,1]$ be a function such that for every $x=(x_1,\ldots,x_n) \in \oo^n$, satisfies $\size{f(x) - \frac1n\sum_{i=1}^n x_i} \leq 1/3$. Then, 
    \begin{align*}
        \eex{p \la [-1,1], \: x_{1\ldots n} \sim p}{f(x) \cdot \sum_{i=1}^n (x_i - p)} \geq \Omega(1),
    \end{align*}
    where $p \la [-1,1]$ denotes that $p$ is sampled uniformly over $[-1,1]$, and $x_{1\ldots n} \sim p$ denotes that each $x_i \in \oo$ is sampled independently with $\ex{x_i} = p$.
\end{lemma}

Roughly, \cref{lem:intro:origFPL} says that if $f(x) \approx \frac1n\sum_{i=1}^n x_i$, then $f(x)$ has $\Omega(1/n)$ correlation (on average) with each $x_i$. In order to increase the correlation, we increase the dimension of the $x_i$'s by using $d = \tilde{\Theta}(n^2)$ independent copies for each coordinate (column). This guarantees that the average correlation that a single word $x_i \in \oo^d$ has with an accurate output $q \in [-1,1]^d$ is $\tilde{\Omega}(d/n)$, which is sufficiently larger than the $\tilde{O}(\sqrt{d})$ correlation that an independent row (which was not part of the input) has with $q$.

However, in our case, since \cref{lem:intro:origFPL} is only restricted to adversaries that must be accurate for all types of columns, and not just marked ones, we could not use it for our padding-and-permuting (PAP) technique, and therefore we developed the following stronger Fingerprinting Lemma.

\begin{lemma}[Our Strong Fingerprinting Lemma]\label{lem:intro:new_FP_lemma}
    Let $f \colon \oo^n \rightarrow [-1,1]$ with $f(1,\ldots,1) = 1$ and $f(-1,\ldots,-1) = -1$, and
	let $\rho$ be the distribution that outputs $p = \frac{e^t - 1}{e^t + 1}$ for $t \la [-\ln(5n),\ln(5n)]$. Then,
	\begin{align*}
		\eex{p \sim \rho, \: x_{1\ldots n} \sim p}{f(x) \cdot \sum_{i =1}^n (x_i - p)} \geq  \Omega(1/\log n).
	\end{align*}
\end{lemma}

Note that up to the $\log n$ factor, \cref{lem:intro:new_FP_lemma} is much stronger than the original fingerprinting lemma that is used in the literature (\cref{lem:intro:origFPL}), since it only requires $f$ to be fixed on the two points, $(1,\ldots,1)$ and $(-1,\ldots,-1)$, and nothing else (i.e., $f$ can be completely arbitrary on any other input).

Now the same approach of \cite{BunSU17,DworkSSUV15} yields that we can increase the correlation by taking $d = \tilde{\Theta}(n^2)$ independent columns,\footnote{An additional factor of $\log^2 n$ is hidden inside the $\tilde{\Theta}$ due to the $\log n$ factor that we lose in our new fingerprinting lemma. However, we ignore it for the sake of this presentation as it is only a low-order term.} which leads to a tracing algorithm. But unlike \cite{BunSU17,DworkSSUV15}, we obtain a much stronger ``FPC-style" tracing algorithm, which enables to apply a similar PAP approach as described in \cref{sec:ProofViaTardos}. As a corollary that is of independent interest, the above approach results in a new fingerprinting code that has a simpler analysis than the one of \cite{Tardos08}, described in \cref{sec:appendix:FPC}.

An additional advantage of \cref{lem:intro:new_FP_lemma} is that it also extends to randomized functions $f$ with $\pr{f(1,\ldots,1) = 1}, \pr{f(-1,\ldots,-1) = -1} \geq 0.9$. This yields an FPC against mechanisms $\Mc$ that given a codebook $X$, the output $\Mc(X)$ is \emph{strongly-correlated} with $X$:

\begin{definition}[Strongly Correlated]\label{def:technique:strongly-correlated}
	We say that a random variable $Q = (Q^1,\ldots,Q^d) \in \oo^d$ is \emph{strongly-correlated} with a matrix $X \in \oo^{n \times d}$, if
	\begin{align*}
		\forall b\in \oo, \: \forall j \in \cJ_X^b: \quad \pr{Q^j = b} \geq 0.9.
	\end{align*}
\end{definition}

Note that \emph{strongly-correlated} (\cref{def:technique:strongly-correlated}) is similar to \emph{strongly-agrees} (\cref{def:intro:strongly-agree}) but not exactly the same (in particular, the former is a property of a random variable while the latter is a property of a fixed vector). In our analysis, we use the \emph{strongly-correlated} definition which makes the statements and proofs much more clean. To see this, recall that in the proof sketch in \cref{sec:ProofViaTardos}, we claimed that an algorithm that agrees with $90\%$ of the marked columns of $X'$, also agrees with a ``similar" fraction of marked columns of the original codebook $X$. But this is not exactly the same agreement fraction, and it creates some additional restrictions on the parameters, and in particular, one has to go into the specific FPC construction to argue that w.h.p., there are many marked columns (as done by \cite{BUV14} w.r.t. \citeauthor{Tardos08}'s code). But in this work we observed that padding-and-permuting simply transforms an $0.9$-agreement on $X'$ (i.e., \emph{strong-agreement}) into $0.9$-correlation on $X$ (i.e., \emph{strong-correlation}) without any special requirements on the parameters (this is the content of \cref{lemma:PAP}), so we were able use our FPC directly. 

%Since we constructed an FPC that handles strongly-correlated mechanisms, we were able to prove cleaner statements.

\subsection{More General Framework}\label{sec:technique:framework}

%\subsubsection{New Framework}

While our tools (\cref{thm:intro:base_tool,thm:intro:extended_tool}) capture various fundamental problems, they may not serve as the right abstraction for handling more complicated output spaces (e.g., matrices) that are used for other problems (e.g., subspace or covariance estimation). We therefore provide a more general framework that we hope can be applied to other types of algorithms without the need to develop similar tools from scratch. 

Consider a mechanism $\Mc \colon \cX^n \rightarrow \cW$ that satisfies some weak accuracy guarantee.
In order to prove a lower bound on $n$ using our approach, we need somehow to transform an FPC codebook $X\in \oo^{n_0 \times d_0}$ into hard instances $Y \in \cX^n$ for $\Mc$, and then extract from the output $w \in \cW$ of $\Mc(Y)$ a vector $q \in \oo^{d_0}$ that is strongly-correlated with $X$ ($n_0$ and $d_0$ are some functions of $n$ and $d$ and the weak accuracy guarantee of $\Mc$). 
Denote by $\Gc \colon \oo^{n_0 \times d_0} \times \cV \rightarrow \cX^n$ the algorithm that generates the hard instances using a uniformly random secret $v \la \cV$ (i.e., $v$ could be a random permutation, a sequence of random permutations, etc). 
Denote by $\Fc \colon \cV \times \cW \rightarrow \oo^{d_0}$ the algorithm that extracts a good $q$ using the secret $v$ and the output $w$.
Denote by $\Ac^{\Mc, \Fc, \Gc}(X)$ the process that samples $v \la \cV$, and outputs $q \sim \Fc(v, \Mc(\Gc(X,v)))$.

Our framework (\cref{lemma:framework}) roughly states that if $\Mc$ is $\paren{1, \frac{\beta}{4n_0}}$-DP and there exists such $\Gc, \Fc$ where: (1) The output of $\Ac^{\Mc, \Fc, \Gc}(X)$ is strongly-correlated with $X$ w.p. at least $\beta$ over the random coins of $\Mc, \Fc, \Gc$, and (2) $\Gc$ is neighboring-preserving (i.e., maps neighboring datasets to neighboring datasets), then $n_0 \geq \Omega\paren{\frac{\sqrt{d_0}}{\log^{1.5}(d_0/ \beta)}}$.

We prove \cref{thm:intro:base_tool} by applying the framework with $n_0 = n$ and $d_0 = \alpha d$, and we prove \cref{thm:intro:extended_tool} by applying it with $n_0 = n/k$ and $d_0 = \alpha d$.

\subsection{Comparison with \cite{NarayananME22}}\label{sec:technique:comparison}

\citet{NarayananME22} developed very different hard-instances for lower bounding \emph{user-level} DP averaging that can be used to prove a similar statement to \cref{thm:intro:lower_bound:avg} (our DP averaging lower bound).
In their setting, a distribution vector $p = (p_1,\ldots,p_d) \in [0.5/d,1.5/d]^d$, $\sum_{i=1}^d p_i = 1$, is sampled, and the goal is to estimate it. For each user (out of $m$), they sample $m = \lambda^2 d$  one-hot vectors according to $p$, and provide the average of the vectors as the user's input point. Since each $p_i \approx 1/d$, it can be shown that w.h.p., the resulting points have $\ell_2$ diameter $\gamma \approx 1/\sqrt{m} = 1/\lambda \sqrt{d}$ (and also close to $p$ up to such an additive error).
So their Lemma 23 essentially states that estimating $p$ up-to additive $\ell_2$ error of $1/\sqrt{d} = \lambda\cdot \gamma$, requires $\tilde{\Omega}(\sqrt{d}/\lambda)$ users (which matches our \cref{thm:intro:lower_bound:avg} up to the low-order terms). 

While both works yield similar smooth lower bound for DP averaging, the focus of the works is different: \cite{NarayananME22} focus on \emph{user-level} DP averaging, while our focus is on general bounds for various (\emph{item-level}) DP problems. As part of that, we have several advantages in the \emph{item-level} DP case:

\begin{enumerate}
	\item Applicability: Our padding-and-permuting FPC hard-instances enable to prove more general tools (\cref{thm:intro:base_tool,thm:intro:extended_tool}) which provide clean abstraction for lower bounding other fundamental problems under DP, like clustering and estimating the top-singular vector.
	
	\item Simplicity:  Our proof is conceptually cleaner and simpler, and in particular, can be explained using a simple black-box construction from an optimal fingerprinting-code (\cref{sec:ProofViaTardos}).
	
	\item Tighter Bounds: \cite{NarayananME22} have an $1/\log^7 (d/\lambda)$  dependency in their DP-averaging lower bound (see their Theorem 10), while we only have an $1/\log^{1.5} (d/\lambda)$ dependency.
	
	\item Boolean instances: We create \emph{boolean} hard-instances, while \cite{NarayananME22}'s instances are not.
\end{enumerate}

\section{Preliminaries}\label{sec:prelim}

\subsection{Notations}
We use calligraphic letters to denote sets and distributions, uppercase for matrices and datasets, boldface for random variables, and lowercase for vectors, values and functions. 
For $n \in \bbN$, let $[n] = \set{1,2,\ldots,n}$. 
Throughout this paper, we use $i \in [n]$ as a row index, and $j \in [d]$ as a column index (unless otherwise mentioned).

For a matrix $X = (x_i^j)_{i\in [n], j \in [d]}$, we denote by $x_i$ the \ith row of $X$ and by $x^j$ the \jth column of $X$. A column vector $x \in \bbR^n$ is written as $(x_1,\ldots,x_n)$ or $x = x_{1\ldots n}$, and a row vector $y \in \bbR^d$ is written as $(y^1,\ldots,y^d)$ or $y^{1\ldots d}$. In this work we consider mechanisms who receive an
$n \times d$ matrix $X$ as input, which is treated as the dataset $X = (x_1,\ldots,x_n)$ where the rows of $X$ are the elements (and therefore, we sometimes write $X \in (\bbR^d)^n$ instead of $X \in \bbR^{n \times d}$ to emphasize it).
For $d \in \bbN$ we denote by $\cP_d$ the set of all $d \times d$ permutation matrices.

For a vector $x \in \bbR^d$ we define $\norm{x}_1 = \sum_{i=1}^n \size{x_i}$ (the $\ell_1$ norm of $x$), and $\norm{x}_2 = \sqrt{\sum_{i=1}^n x_i^2}$ (the $\ell_2$ norm of $x$), and for a subset $\cS \subseteq [d]$ we define $x_{\cS}=(x_i)_{i \in \cS}$, and in case $x$ is a row vector we write $x^{\cS}$. Given two vectors $x = (x_1,\ldots,x_n), y = (y_1,\ldots,y_n)$, we define $\ip{x,y} = \sum_{i=1}^n x_i y_i$ (the inner-product of $x$ and $y$). 
For a matrix $X = (x_i^j)_{i\in [n], j \in [d]} \in \oo^{n \times d}$ and $b \in \oo$, we define the $b$-marked columns of $X$ as the subset $\cJ^b_X \subseteq [d]$ defined by $\cJ^b_X = \set{j \in [d] \colon x_i^j = b \text{ for all }i\in[n]}$.

For $z \in \bbR$,  we define $\signn{z} \eqdef \begin{cases} 1 & z \geq 0 \\ -1 & z<0\end{cases}$ and for $v = (v^1,\ldots,v^d) \in \bbR^d$ we define $\signn{v} \eqdef (\signn{v^1},\ldots,\signn{v^d}) \in \oo^d$.

\subsection{Distributions and Random Variables}\label{sec:prelim:dist}

Given a distribution $\cD$, we write $x \sim \cD$ to denote that $x$ is sampled according to $\cD$.
For a set $\cS$, we write $x \la \cS$ to denote that $x$ is sampled from the uniform distribution over $\cS$.
%For a random variable $\bx$ we write $\bx \sim \cD$ or $\bx \la \cS$ to denote the same things as above.

\begin{fact}[Hoeffding's inequality \cite{hoeff}]
%Hoeffding's inequality. 
Let $\bx_1,\dots, \bx_n$ be independent random variables taking integer values in the range $[a,b]$. 
Also, let $\bx=\sum_{i=1}^n \bx_i$ denote the sum of the variables and $\mu = \eex{}{\bx}$ denote its expectation.
Then, for any $t>0$,
\begin{align}
    \pr{\bx \leq \mu - t} \leq e^{-\frac{2t^2}{n(b-a)^2}}.\label{hoeffding_ineq1}
\end{align}

\begin{align}
    \pr{\bx \geq \mu + t} \leq e^{-\frac{2t^2}{n(b-a)^2}}.\label{hoeffding_ineq2}
\end{align}
\end{fact}

\begin{definition}[Behave the same]\label{def:behave_the_same}
	We say that two random variables $\bx$ and $\bx'$ over $\cX$ behave the same w.p. $\beta$, if there exists a random variable $\by$ over $\cY$ (jointly distributed with $\bx,\bx'$), and event $E \subseteq \cY$ with $\pr{\by \in E} \geq \beta$ such that $\bx|_{\by \in E} \equiv \bx'|_{\by \in E}$.
\end{definition}

\begin{fact}\label{fact:behave_the_same}
	If $\bx$ and $\bx'$ behave the same w.p. $\beta$, then for any event $F$, $$\pr{\bx \in F} \geq \pr{\bx' \in F} - (1 - \beta).$$
\end{fact}
\begin{proof}
	Let $\by,E$ as in \cref{def:behave_the_same}. Compute
	\begin{align*}
		\pr{\bx \in F} 
		&\geq \pr{\by \in E} \cdot \pr{\bx \in F \mid \by \in E}\\
            &= \pr{\by \in E} \cdot \pr{\bx' \in F \mid \by \in E}\\
		&\geq  \pr{\by \in E} \cdot \frac{\pr{\bx' \in F} - \pr{\by \notin E}}{\pr{\by \in E}}\\
		&= \pr{\bx' \in F} - \pr{\by \notin E} \geq \pr{\bx' \in F} - (1 -\beta).
	\end{align*}
\end{proof}

\subsection{Algorithms}

Let $\Mc$ be a randomized algorithm that uses $m$ random coins. For $r \in \zo^m$ we denote by $\Mc_r$ the (deterministic) algorithm $\Mc$ after fixing its random coins to $r$.
We use the same notation for more specific cases, e.g., if the random choices of $\Mc$ consist of sampling $s \la [k]$ and $P \la \cP_d$, then for $s\in[k]$ and $P \in \cP_d$ we denote by $\Mc_{s,P}$ the algorithm $\Mc$ after fixing this random choices to $s$ and $P$ (respectively).

\subsection{Differential Privacy}

\begin{definition}[Differential Privacy~\citep{DMNS06,DKMMN06}]\label{def:dp} 
    A randomized mechanism $\Mc \colon \cX^n \rightarrow \cY$ is \emph{$(\eps,\delta)$-differentially private} (in short, $(\eps,\delta)$-DP) if for every neighboring databases $X=(x_1,\ldots,x_n), \: X' = (x_1',\ldots,x_n') \in \cX^n$ (i.e., differ by exactly one entry), and every set of outputs $\cT \subseteq \cY$, it holds that
    \begin{align*}
        \pr{\Mc(X) \in \cT} \leq e^{\eps} \cdot \pr{\Mc(X') \in \cT} + \delta
    \end{align*}
\end{definition}

\subsubsection{Known Facts}

\begin{fact}[Post-Processing]\label{fact:post-processing}
    If $\Mc \colon \cX^n \rightarrow \cY$ is $(\eps,\delta)$-DP then for every randomized $\Fc \colon \cY \rightarrow \cZ$, the mechanism $\Fc\circ \Mc \colon \cX^n \rightarrow \cZ$ is $(\eps,\delta)$-DP.
\end{fact}

Post-processing holds when applying the function on the output of the DP mechanism.
In this work we sometimes need to apply the mechanism on the output of a function. While this process does not preserve DP in general, it does so assuming the function is \emph{neighboring-preserving}.  

\begin{definition}[Neighboring-Preserving Algorithm]\label{def:neighbor-preserving}
	We say that a randomized algorithm $\Gc \colon \cX^n \rightarrow \cY^m$ is \emph{neighboring-preserving} if for every neighboring $X, X' \in \cX^n$, the outputs $\Gc(X), \Gc(X') \in \cY^m$ are neighboring w.p. $1$.
\end{definition}

\begin{fact}\label{fact:neighboring-to-neighboring}
	If $\Gc \colon \cX^n \rightarrow \cY^m$ is neighboring-preserving and $\Mc \colon \cY^m \rightarrow \cZ$ is $(\eps,\delta)$-DP, then $\Mc \circ \Gc \colon \cX^n \rightarrow \cZ$ is $(\eps,\delta)$-DP.
\end{fact}

We also use the following fact which states that the output of a DP mechanism cannot be too correlated with one of the input points. 

\begin{fact}\label{fact:predicate}
    Let $\Mc \colon \cX^n \rightarrow \cY$ be an $(\eps,\delta)$-DP mechanism, let $i \in [n]$,  
    let $\bx_1,\ldots,\bx_n,\bx_i'$ be i.i.d. random variables over $\cX$, let $\bX = (\bx_1,\ldots,\bx_n)$, and let $P \colon \cX \times \cY \rightarrow \zo$ be a (possibly randomized) predicate. Then,
    \begin{align*}
        \pr{P(\bx_i, \Mc(\bX)) = 1} \leq e^{\eps} \cdot \pr{P(\bx_i',\Mc(\bX)) = 1} + \delta.
    \end{align*}
\end{fact}
\begin{proof}
    We assume \wlg that $P$ is deterministic since the proof for a randomized $P$ follows by fixing the random coins of $P$.
    In the following, for $x_i \in \cX$ let $\cT_{x_i} = \set{y \in \cY \colon P(x_i,y) = 1}$.
    Compute
    \begin{align*}
        \pr{P(\bx_i,\Mc(\bX)) = 1}
        &= \eex{x_i \sim \bx_i}{\pr{\Mc(\bX\mid_{\bx_i = x_i}) \in \cT_{x_i}}}\\
        &= \eex{x_i \sim \bx_i'}{\pr{\Mc(\bX\mid_{\bx_i = x_i}) \in \cT_{x_i}}}\\
        &\leq \eex{x_i \sim \bx_i'}{e^{\eps}\cdot \pr{\Mc(\bX) \in \cT_{x_i}} + \delta}\\
        &= e^{\eps} \cdot \pr{P(\bx_i',\Mc(\bX)) = 1} + \delta,
    \end{align*}
    where the second equality holds since $\bx_i \equiv \bx_i'$, and the inequality holds since $\Mc$ is $(\eps,\delta)$-DP and since $\bX\mid_{\bx_i=x_i}$ and $\bX$ are neighboring. %w.p. $1$.
\end{proof}

\section{Strong Fingerprinting Lemma}\label{sec:FPL}

In this section, we prove our new fingerprinting lemma. 
In \cref{sec:FPL:1dim} we prove \cref{lem:intro:new_FP_lemma} and extend it to a robust version that allows a small error probability (\cref{lemma:robust_FPL}). In \cref{sec:FPL:HighDim} we increase the dimension and prove the resulting correlation in \cref{lemma:FPL_high_dim}.

In this section we adopt the notation of \citet{DworkSSUV15}.
For $p \in [-1,1]$ let $x \sim p$ denote that $x \in \set{\pm 1}$ is drawn with $\ex{x} = p$, and let $x_{1\ldots n} \sim p$ denote that each $x_i$ is sampled independently with $\ex{x_i} = p$ (that is, $\pr{x_i = 1} = \frac{1+p}{2}$ and $\pr{x_i = -1} = \frac{1-p}{2}$).

\subsection{One Dimension}\label{sec:FPL:1dim}
In this section, we prove \cref{lem:intro:new_FP_lemma}. We make use of the following lemma.

\begin{lemma}[\cite{DworkSSUV15}, Lemma 5]\label{lem:Dwork}
	Let $f \colon \oo^n \rightarrow \bbR$. Define $g \colon [-1,1] \rightarrow \bbR$ by 
	\begin{align*}
		g(p) \eqdef \eex{x_{1\ldots,n} \sim p}{f(x)}.
	\end{align*}
	Then
	\begin{align*}
		\eex{x_{1\ldots n} \sim p}{f(x) \cdot \sum_{i \in [n]} (x_i - p)} = g'(p)\cdot (1-p^2).
	\end{align*}
\end{lemma}

Note that when $f(x) \approx \frac1n \sum_{i=1}^n x_i$, then $g(p) \approx p$.
Therefore, if we choose $p$ uniformly over $[-1,1]$, then \cref{lem:Dwork} implies that we get $\Omega(1)$-advantage, which results in the original fingerprinting lemma (\cref{lem:intro:origFPL}). However, in \cref{lem:intro:new_FP_lemma}, we are interested in a much weaker $f$ that only satisfies $f(1,\ldots,1) = 1$ and $f(-1,\ldots,-1) = -1$, and might be arbitrary in all other inputs in $\oo^n$. So the only thing that we know about such $f$ is that it induces a function $g \colon [-1,1] \rightarrow [-1,1]$ such that $g(-1) = -1$ and $g(1) = 1$, and actually we can also show that $g(-1+\epsilon) \approx -1$ and $g(1 - \epsilon) \approx 1$ for $\epsilon = O(1/n)$. We show that these limited properties suffice for a fingerprint lemma, by choosing $p$ from a distribution $\rho$ that has probability density function $\propto 1/(1-p^2)$.

\begin{lemma}[Our Strong Fingerprinting Lemma, Restatement of \cref{lem:intro:new_FP_lemma}]\label{new_FP_lemma}
	Let $f \colon \oo^n \rightarrow [-1,1]$ with $f(1,\ldots,1) = 1$ and $f(-1,\ldots,-1) = -1$, and
	let $\rho$ be the distribution that outputs $p = \frac{e^t - 1}{e^t + 1}$ for $t \la [-\ln(5n),\ln(5n)]$. Then,
	\begin{align*}
		\eex{p \sim \rho, \: x_{1\ldots n} \sim p}{f(x) \cdot \sum_{i \in [n]} (x_i - p)} \geq  \frac1{\ln(5n)}.
	\end{align*}
\end{lemma}

\begin{proof}
	Let $\bp$ denote a random variable that is distributed according to $\rho$. Let's first compute its cumulative distribution function (CDF):
	\begin{align*}
		\pr{\bp \leq p} 
		&= \ppr{t \la [-\ln(5n),\ln(5n)]}{\frac{e^t - 1}{e^t + 1} \leq p}\\
		&= \ppr{t \la [-\ln(5n),\ln(5n)]}{t \leq \ln \paren{\frac{1+p}{1-p}}}\\
		&=  \frac{\ln \paren{\frac{1+p}{1-p}} + \ln(5n)}{2 \ln (5n)}\\
		%&= \frac1{2\ln(5n)}\cdot \ln \paren{\frac{1+p}{1-p}} + \frac{1}{2}\\
            &= \frac1{2\ln(5n)}\cdot \paren{\ln (1+p) - \ln (1-p)} + \frac{1}{2}
	\end{align*}
 
	Hence, the probability density function (PDF) of $\bp$ is 
	\begin{align*}
		\rho(p) 
		&\eqdef \frac{d \pr{\bp \leq p}}{dp}\\
		&= \frac1{2 \ln (5n)}\cdot \paren{\frac1{1+p} + \frac1{1-p}}\\
		%&= \frac1{2 \ln (5n)}\cdot \frac{1-p}{1+p} \cdot \frac2{(1-p)^2}\\
		&= \frac1{\ln (5n)}\cdot \frac{1}{1-p^2}.
	\end{align*}
	%where in the first equality we used the fact that $\frac{d \ln h(p)}{dp} = \frac{1}{h(p)} \cdot \frac{d h(p)}{dp}$.
	
	In the following, let $g(p) \eqdef \eex{x_{1\ldots n} \sim p}{f(x)}$.
	Note that $t = \ln (5n) \implies p = 1 - \frac{2}{5n+1}$ and $t = -\ln(5n) \implies p = -1 +\frac{2}{5n+1}$.
	In addition, note that if $x_i \sim p$, then $\pr{x_i = 1} = \frac{1+p}{2}$ and $\pr{x_i = -1} = \frac{1-p}{2}$.	
	By the assumption on $f$, for every $p$ it holds that
	\begin{align*}
		g(p) \geq \ppr{x_{1\ldots n} \sim p}{\forall i: x_i = 1} - \ppr{x_{1\ldots n} \sim p}{\exists i: x_i = -1} = 2\cdot \paren{\frac{1+p}{2}}^n - 1
	\end{align*}
	and
	\begin{align*}
		g(p) \leq -\ppr{x_{1\ldots n} \sim p}{\forall i: x_i = -1} + \ppr{x_{1\ldots n} \sim p}{\exists i: x_i = 1} = -2\cdot \paren{\frac{1-p}{2}}^n + 1
	\end{align*}
	Hence,
	\begin{align*}
		g\paren{1 - \frac{2}{5n+1}} \geq 2\cdot \paren{1 - \frac1{5n+1}}^n - 1 \geq 2\cdot e^{-\frac{1.1n}{5n+1}} - 1 \geq 2\cdot e^{-1/4} - 1 \geq 0.5
	\end{align*}
        and
	\begin{align*}
		g\paren{-1+\frac{2}{5n+1}} \leq -2\cdot \paren{1 - \frac1{5n+1}}^n + 1 \leq -0.5.
	\end{align*}
    where we used the inequality $e^{-1.1 y} \leq 1-y$ for every $y \in [0,1/6]$. 
	
	Hence, we conclude that
	\begin{align}
		\eex{p \sim \rho, \: x_{1\ldots n} \sim p}{f(x) \cdot \sum_{i \in [n]} (x_i - p)}
		&= 	\eex{p \sim \rho}{g'(p)\cdot (1-p^2)}\nonumber\\
		&= \int_{-1+\frac{2}{5n+1}}^{1-\frac{2}{5n+1}} g'(p)\cdot (1-p^2) \cdot \rho(p) dp\nonumber\\
		&= \frac1{\ln (5n)} \cdot \int_{-1+\frac{2}{5n+1}}^{1-\frac{2}{5n+1}} g'(p) dp\nonumber\\
		&= \frac1{\ln (5n)} \cdot \paren{g\paren{1-\frac{2}{5n+1}} - g\paren{-1+\frac{2}{5n+1}}}\label{g_minus_g}\\
		&\geq \frac1{\ln (5n)}.\nonumber
	\end{align}
\end{proof}

\begin{remark}
In \cref{new_FP_lemma}, it is possible to increase the range from which we choose the parameter $t$ in order to slightly improve the lower bound on the expectation $\eex{}{f(x) \cdot \sum_{i \in [n]} (x_i - p)}$; i.e., as we take larger range, we will get better bound (for large enough values of $n$).
For example, taking $t \la [-\ln(20n),\ln(20n)]$ will give us $\eex{}{f(x) \cdot \sum_{i \in [n]} (x_i - p)} \geq  \frac{1.75}{\ln(20n)}$.
We choose the range to be $[-\ln(5n),\ln(5n)]$ since it suffices for us in order to achieve positive expectation in \cref{lemma:robust_FPL}.
\end{remark}

\begin{lemma}[Robust version of \cref{new_FP_lemma}]\label{lemma:robust_FPL}
	Let %$q \in [0,1]$ and 
    $F \colon \oo^n \rightarrow [-1,1]$ be a randomized function such that $\pr{F(1,\ldots,1) = 1} \geq 0.9$ and $\pr{F(-1,\ldots,-1) = -1} \geq 0.9$. 
	Then,
	\begin{align*}
		\eex{p \sim \rho, \: x_{1\ldots n} \sim p}{F(x) \cdot \sum_{i \in [n]} (x_i - p)} \geq  \frac{0.4}{\ln (5n)}
	\end{align*}
	for $\rho$ as in \cref{new_FP_lemma}, where the expectation is also over the randomness of $F$.
\end{lemma}

\begin{proof}
Assume that $F(x)$ is defined by $f(x;r)$ for a random $r \la \zo^m$. 
Let $q =0.9$ and  $\cR = \set{r \in \zo^m \colon f(1,\ldots,1;r) = 1 \land f(-1,\ldots,-1;r) = -1}$.  By assumption and the union bound, $\size{\overline{\cR}} \leq 2(1-q)\cdot 2^m =(2-2q)\cdot 2^m$, so $\size{\cR} = 2^m - \size{\overline{\cR}} \geq (2q-1) \cdot 2^m$.
 
 By \cref{new_FP_lemma},
	\begin{align*}
		\forall r \in \cR: \quad \eex{p \sim \rho, \: x_{1\ldots n} \sim p}{f(x; r) \cdot \sum_{i \in [n]} (x_i - p)} \geq  \frac1{\ln(5n)}.
	\end{align*}
By \cref{g_minus_g} and since for every $r$, the function $f(\cdot ,r)$ induces a function $g$ with range $[-1,1]$, it holds that
	\begin{align*}
		\forall r \notin \cR:  \quad 
		\eex{p \sim \rho, \: x_{1\ldots n} \sim p}{f(x; r) \cdot \sum_{i \in [n]} (x_i - p)}
		&= \frac1{\ln (5n)} \cdot \paren{g\paren{1-\frac{2}{5n+1}} - g\paren{-1+\frac{2}{5n+1}}}\\
		&\geq -\frac{2}{\ln (5n)}.
	\end{align*} 
 
Therefore, we conclude that
	\begin{align*}
		\eex{p \sim \rho, \: x_{1\ldots n} \sim p}{F(x) \cdot \sum_{i \in [n]} (x_i - p)}
		&= \eex{r \la \zo^m, \: p \sim \rho, \: x_{1\ldots n} \sim p}{f(x; r) \cdot \sum_{i \in [n]} (x_i - p)}\\
		&\geq (2q-1) \cdot \frac1{\ln (5n)} - (2-2q) \cdot \frac{2}{\ln (5n)}\\
		&= \frac{6q-5}{\ln (5n)}\\
            &= \frac{0.4}{\ln (5n)}.
		%&\geq \frac{0.04}{\ln n}
	\end{align*}
\end{proof}

\subsection{Increasing the Dimension}\label{sec:FPL:HighDim}

In this section, we increase the correlation achieved by our fingerprinting lemma by increasing the number of columns (the dimension of the input points). 

We start by defining a \emph{strongly-accurate} mechanism, which is a natural extension of the one-dimension case. 

\begin{definition}[Strongly-Accurate Mechanism]\label{def:strong_acc}
	We say that a mechanism $\Mc \colon \oo^{n \times d} \rightarrow [-1,1]^{d}$ is \emph{strongly-accurate} if for every $X \in \oo^{n \times d}$, $\Mc(X)$ is strongly-correlated with $X$ (\cref{def:technique:strongly-correlated}).
\end{definition}
That is, for every $b$-marked column $j \in [d]$, a strongly-accurate mechanism returns $b$ as the \jth element of its output with high probability. 

We next define our hard distribution.

\begin{definition}[Distributions $\cD'(n,d)$ and $\cD(n,d)$]\label{def:D}
	Let $\rho$ be the distribution from \cref{new_FP_lemma}. Define $\cD'(n,d)$ to be the distribution that outputs  $(x_1,\ldots,x_n, z)\in \oo^{(n+1) \times d}$ where $p^1, \ldots, p^d \sim \rho$ are sampled (independently), each $x_i^j \in \oo$ is sampled with $\ex{x_i^j} = p^j$, and each $z^j$ is sampled with $\ex{z^j} = p^j$.
	Define $\cD(n,d)$ as the distribution of the first $n$ vectors in $\cD'(n,d)$ (i.e., without $z$).
\end{definition}

\begin{lemma}\label{lemma:FPL_high_dim}
	Let $\Mc \colon \oo^{n \times d} \rightarrow [-1,1]^d$ be a \emph{strongly-accurate} mechanism, and let $(\bx_1,\ldots,\bx_n,\bz) \sim \cD'(n,d)$ (\cref{def:D}), $\bX = (\bx_1,\ldots,\bx_n)$ and $\bq = \Mc(\bX)$ be random variables. Then,
	\begin{align*}
		\ex{\sum_{i=1}^n \paren{\ip{\bx_i, \bq} - \ip{\bz,\bq}}} \geq \frac{0.4 d}{\ln (5n)}.
	\end{align*}
	Furthermore, for every $\beta \in [0,1]$, if $d \geq \Theta(n^2\log^2 (n) \log(1/\beta))$, then
 
\begin{align*}
		\pr{\sum_{i=1}^n \paren{\ip{\bx_i, \bq} - \ip{\bz, \bq}} \leq \frac{0.2 d}{\ln (5n)}} \leq \beta.
	\end{align*}
\end{lemma}
\begin{proof}
The first part of the lemma follows from similar arguments as in the proof of Proposition 10 in \cite{DworkSSUV15}.
In particular, for a given $j \in [d]$, we fix all columns except for the \jth column.
Since the mechanism $\Mc$ only gets $\bX$ as an input, and each columns $j$ of $\bX$ is sampled independently, then 
there is some $F_j \colon \oo^{n} \rightarrow [-1,1]^d$ for which $\bq^j \sim F_j(\bx^j)$, where $\bx^j=(\bx_1^j,\dots,\bx_n^j)$ (the $j$'th column of $\bX$), such that \cref{lemma:robust_FPL} holds for it.
Then, by this and the fact that $\Mc$ is a strongly-accurate mechanism,
\begin{align*}
	\ex{\bq^j \cdot \sum_{i \in [n]} (\bx_i^j - \bz^j)} 
	= \ex{F_j(\bx^j) \cdot \sum_{i \in [n]} (\bx_i^j - \bp_j)}
	\geq  \frac{0.4}{\ln (5n)},
\end{align*}
where $\bp^1,\ldots,\bp^d$ are the expectations that were sampled in the process of sampling $\bX,\bz$ (\cref{def:D}).
Thus, by linearity of expectation,  we conclude the first part of the lemma.

The second part is proven similarly to Lemma 11 in \cite{DworkSSUV15}.
Let us define the random variable $\ba_j = \bq^j \cdot \sum_{i \in [n]} (\bx_i^j - \bz^j)$ for every $j \in [d]$. 
Then, we get that
$$\sum_{i=1}^n \paren{\ip{\bx_i, \bq} - \ip{\bz,\bq}} = \sum_{j=1}^d \bq^j \sum_{i=1}^n (\bx_i^j - \bz^j) = \sum_{j=1}^d \ba_j,$$
and from the first part we have $\ex{\sum_{j=1}^d \ba_j} \geq \frac{0.4d}{\ln (5n)}$.
In order to use Hoeffding's inequality on the sum $\sum_{j=1}^d \ba_j$, we first assume that the random variables $\ba_1,\dots,\ba_d$ are independent. Later, as in \cite{DworkSSUV15}, we show how to remove this assumption.
Next, observe that $-2n \leq \size{\ba_j} \leq 2n$ for every $j \in [n]$. 
Therefore, using an Hoeffding's inequality (\ref{hoeffding_ineq1}), it follows that 
\begin{align*}
    \ppr{}%(x_1,\ldots,x_n, z) \sim \cD, \: q=(q^1,\ldots,q^d) \sim \Mc(x_1,\ldots,x_n)
    {\sum_{i=1}^n \paren{\ip{\bx_i, \bq} - \ip{\bz,\bq}} \leq \frac{0.2 d}{\ln (5n)}} 
    = \ppr{}{\sum_{j=1}^d \ba_j \leq \frac{0.2 d}{\ln (5n)}}
    \leq e^{- \frac{2(0.2d/\ln(5n))^2}{d (4n)^2}}
    = e^{\Theta\big(- \frac{d}{n^2 \ln^2 (n)}\big)}
    \leq \beta.
\end{align*}
Now, for the case that $\ba_1,\dots,\ba_d$ are not independent, we use the fact that the sum $\sum_{j=1}^d \ba_j$ concentrates as if they were independent, which follows since for every $j \in [d]$, the expected value of $\ba_j$ is the same even when conditioned on the other variables of $\set{\ba_1,\dots,\ba_d}\setminus \ba_j$. That is, 
\begin{align*}
    \ex{\ba_j} = \ex{\ba_j \;\middle|\; \set{\ba_1,\dots,\ba_d}\setminus \ba_j}.
\end{align*}
\end{proof}

\section{Framework for Lower Bounds}\label{sec:framework}

Following  \cref{sec:technique:framework}, in this section we present our general framework for lower bounding DP algorithms (stated in \cref{lemma:framework}) that is based on our new hard-distribution.

\begin{definition}[Algorithm $\Ac^{\Mc, \Fc, \Gc}$]\label{def:A_MFG}
	Let $\cV$, $\cW$ be domains, and let $n_0, d_0, n \in \bbN$.
	Let $(\Mc, \Fc, \Gc)$ be a triplet of randomized algorithms of types
	$\: \Gc \colon \oo^{n_0 \times d_0} \times \cV \rightarrow \cX^n$, $\: \Mc \colon \cX^n \rightarrow \cW$, and $\: \Fc \colon \cV \times \cW \rightarrow [-1,1]^{d_0}$, each uses $m$ random coins. 
	Let $\Ac^{\Mc, \Fc, \Gc} \colon \oo^{n_0 \times d_0} \rightarrow [-1,1]^{d_0}$ be the randomized algorithm that on input $X \in \oo^{n_0 \times d_0}$, samples  $v \la \cV$, $Y \sim \Gc(X,v)$ and $w \sim \Mc(Y)$, and outputs $q \sim \Fc(v,w)$.
\end{definition}

\begin{definition}[$\beta$-Leaking]\label{def:beta-leaking}
	
	Let $\Mc, \Fc, \Gc$ be randomized algorithms as in \cref{def:A_MFG}, each uses $m$ random coins, and let $\cD(n_0, d_0)$ be the distribution from \cref{def:D}.
	We say that the triplet $(\Mc, \Fc, \Gc)$ is \emph{$\beta$-leaking} if
	\begin{align*}
		\ppr{r,r', r'' \la \zo^m, \: X \la \cD(n_0,d_0)}{\Ac^{\Mc_{r}, \Fc_{r'}, \Gc_{r''}}(X) \text{ is strongly-correlated with }X} \geq \beta,
	\end{align*}
	(where recall that $\Mc_{r}$ denotes the algorithm $\Mc$ when fixing its random coins to $r$, and $\Fc_{r'}, \Gc_{r''}$ are similarly defined).
\end{definition}

In the following, fix a $\beta$-leaking triplet $(\Mc,\Fc,\Gc)$, and let $\Ac = \Ac^{\Mc, \Fc, \Gc}$ and $\Ac_{r,r',r''} = \Ac^{\Mc_r, \Fc_{r'}, \Gc_{r''}}$.

\begin{claim}\label{claim:M-to-A-DP}
	If $\Mc$ is $(\eps, \delta)$-DP, and $\Gc(\cdot, v)$ is neighboring-preserving (\cref{def:neighbor-preserving}) for every $v \in \cV$, then $\Ac$ is $(\eps, \delta)$-DP.
\end{claim}
\begin{proof}
	Fix $v \in \cV$, and let $\Gc_v = \Gc(\cdot, v)$ and $\Fc_v = \Fc(v, \cdot)$. Since $\Gc_v$ is neighboring-preserving and $\Mc$ is $(\eps,\delta)$-DP, then $\Mc \circ \Gc_v$ is $(\eps,\delta)$-DP (\cref{fact:neighboring-to-neighboring}). By post-processing, $\Fc_v \circ \Mc \circ \Gc_v$ is also $(\eps,\delta)$-DP. Since the above holds for any $v \in \cV$ and since $\Ac  = (\Fc_v \circ \Mc \circ \Gc_v)_{v \la \cV}$, we conclude that $\Ac$ is $(\eps,\delta)$-DP.
\end{proof}

\begin{claim}\label{claim:low_accur_prop:general}
	Let $(\bx_1,\ldots,\bx_{n_0},\bz) \sim \cD(n_0,d_0)$  (\cref{def:D}) and let $\bX = (\bx_1,\ldots,\bx_{n_0}) \in \oo^{n_0 \times d_0}$.
	Assuming $d_0 \geq \Theta(n_0^2 \log^2(n_0) \log(1/\beta))$, it holds that 
	\begin{align*}
		\pr{\sum_{i=1}^{n_0} \paren{\ip{\bx_i, \Ac(\bX)} - \ip{\bz, \Ac(\bX)}} \geq \frac{0.2 d_0}{\ln(5n_0)}} \geq \beta/2.
	\end{align*}
\end{claim}
\begin{proof}
	Consider the algorithm $\Ac'\colon \oo^{n_0 \times d_0} \rightarrow [-1,1]^{d_0}$ that on input $X \in \oo^{n_0 \times d_0}$, 
	samples $r, r', r'' \la \zo^m$ and checks if $\Ac_{r, r', r''}(X)$ is strongly-correlated with $X$. If it does, it outputs $q \sim \Ac_{r, r', r''}(X)$. Otherwise, it outputs $q = (q^1,\ldots,q^{d_0})$ such that $q^j = b$ for every $b \in \oo$ and $j \in \cJ_{X}^b$.
	By definition, $\Ac'$ is strongly-accurate. Therefore, by \cref{lemma:FPL_high_dim} we obtain that
	\begin{align*}
		\pr{\sum_{i=1}^{n_0} \paren{\ip{\bx_i, \Ac'(\bX)} - \ip{\bz, \Ac'(\bX)}} \geq \frac{0.2 d_0}{\ln(5n_0)}} \geq 1 - \beta/2.
	\end{align*}
	But recall that $\Ac = \Ac^{\Mc, \Fc, \Gc}$, where $(\Mc, \Fc, \Gc)$ is $\beta$-leaking (\cref{def:beta-leaking}). Therefore, 
	$\Ac'(\bX)$ behaves as $\Ac(\bX)$ with probability at least $\beta$ (\cref{def:behave_the_same}). Thus by \cref{fact:behave_the_same}
	\begin{align*}
		\pr{\sum_{i=1}^{n_0} \paren{\ip{\bx_i, \Ac(\bX)} - \ip{\bz, \Ac(\bX)}} \geq \frac{0.2 d_0}{\ln(5n_0)}} \geq (1 - \beta/2) - (1 - \beta) = \beta/2.
	\end{align*}
\end{proof}

\begin{claim}\label{claim:large_d0_not_DP}
	If $d_0 \geq  \Theta(n_0^2 \log^2(n_0) \log(n_0/\beta))$, then $\Ac$ is not $\paren{1, \frac{\beta}{4n_0}}$-DP.
\end{claim}
\begin{proof}
	Let $(\bx_1,\ldots,\bx_{n_0},\bz) \sim \cD(n_0,d_0)$, let $\bX = (\bx_1,\ldots,\bx_{n_0})$, and let $\bq = \Ac(\bX)$.
	
	By \cref{claim:low_accur_prop:general} and the union bound, we deduce that there exists $i \in [n_0]$ such that
	\begin{align}\label{eq:event1:general}
		\pr{\ip{\bx_i, \bq} - \ip{\bz,\bq} > \frac{0.2 d_0}{n_0 \ln(5n_0)}} \geq \frac{\beta}{2n_0}. 
	\end{align}
	
	In the following, let $\bp_1,\ldots,\bp_{d_0} \in [-1,1]$ be the (random variables of the) expectations that were chosen as part of the sampling of $(\bx_1,\ldots,\bx_{n_0}, \bz)$ from $\cD(n_0,d_0)$ (as defined in \cref{def:D}), and in the rest of the proof we fix $(\bp_1,\ldots,\bp_{d_0})=(p_1,\ldots,p_{d_0})$ such that \cref{eq:event1:general} holds under this fixing. 
	
	For the above $i \in [n_0]$, let $\bx_i' \in \oo^{d_0}$ be an independent random variable such that each $(\bx_i')^{j} \sim p_j$ (independently) for every $j \in [d_0]$. Note that $\bx_1,\ldots,\bx_{n_0},\bx_i', \bz$ are i.i.d. random variables. However, while $\bq$ depends on $\bx_i$, it is independent of $\bx_i'$ and $\bz$.
	
	Assume towards a contradiction that $\Ac$  is $\paren{1, \frac{\beta}{4n_0}}$-DP.
	By applying \cref{fact:predicate} with the random predicate $P(x_i,q) = \indic{\ip{x_i,q} - \ip{\bz,q} > \frac{0.2 d_0}{n_0 \ln(5n_0)}}$, we deduce from \cref{eq:event1:general} that 
	
	\begin{align}\label{eq:event1:neighbor:general}
		\pr{\ip{\bx_i', \bq} - \ip{\bz,\bq} > \frac{0.2 d_0}{n_0 \ln(5n_0)}} \geq e^{-1}\cdot \paren{\frac{\beta}{2n_0} - \frac{\beta}{4n_0}} \geq \frac{\beta}{4e n_0}. %=\Theta\Big(\frac1{n}\Big).\label{eq:event1}
	\end{align}
	
	In the following, we prove that since $\bx_i'$ is independent of $\bq$, then the above probability is much smaller, which will lead to a contradiction.
	
	For every $j \in [d_0]$, define the random variable $\bb_j = \bq^{j} ((\bx_i')^j - \bz^j)$.
	Then, we get that
	\begin{align*}
		\ip{\bx_i', \bq} - \ip{\bz,\bq} = \sum_{j=1}^{d_0} \bq^{j} ((\bx_i')^j - \bz^j) = \sum_{j=1}^{d_0} \bb_j.
	\end{align*}
	Since $\bx_i', \bz$ and $\bq$ are independent, the above equality yields that $\ex{\sum_{j=1}^{d_0} \bb_j} = 0$.
	Hence, similarly to the proof of \cref{lemma:FPL_high_dim} and since $-2 \leq |\bb_j| \leq 2$ for every $j \in [n_0]$,
	we obtain by Hoeffding's inequality~(\ref{hoeffding_ineq2}) that  
	\begin{align}\label{eq:event2:general}
		\pr{\ip{\bx_i', \bq} - \ip{\bz,\bq} > \frac{0.2 d_0}{n_0 \ln (5n_0)}}
		= \pr{\sum_{j=1}^{d_0} \bb_j > \frac{0.2 d_0}{n_0 \ln (5n_0)}}
		\leq e^{-\frac{d_0}{200 n_0^2 \ln^2 (5 n_0)}}
		\leq \frac{\beta}{20n_0},
	\end{align}
	where the last inequality holds since by assumption $d_0 \geq 200  n_0^2 \ln^2(20n_0)\log\paren{\frac{20 n_0}{\beta}}$ (holds when choosing a large enough constant in the $\Theta$ expression).
	This is in contradiction to \cref{eq:event1:neighbor:general}, so we conclude that $\Ac$ is not $\paren{1, \frac{\beta}{4n_0}}$-DP.
\end{proof}

We now ready to state and prove the guarantee of our framework.

\begin{lemma}[Framework for Lower Bounds]\label{lemma:framework}
	Let $\beta \in (0,1]$, $n_0, n,d_0 \in \bbN$.
	Let $\Mc \colon \cX^n \rightarrow \cW$ be an algorithm such that there exists two algorithms  $\Gc \colon \oo^{n_0 \times d_0} \times \cV \rightarrow \cX^n$ and $ \Fc \colon \cV \times \cW \rightarrow [-1,1]^{d_0}$ such that the triplet $(\Mc, \Fc, \Gc)$ is $\beta$-leaking (\cref{def:beta-leaking}). If $\Mc$ is $\paren{1, \frac{\beta}{4n_0}}$-DP, then $n_0 \geq \Omega\paren{\frac{\sqrt{d_0}}{\log^{1.5}(d_0/ \beta)}}$.
\end{lemma}
\begin{proof}
	Let $c > 0$ be the hidden constant in the $\Theta$ expression of \cref{claim:large_d0_not_DP}. If $d_0 \geq c\cdot n_0^2 \log^2(n_0) \log(n_0/\beta)$, then $\Ac = \Ac^{\Mc,\Fc,\Gc}$ is not $\paren{1, \frac{\beta}{4n_0}}$-DP (\cref{claim:large_d0_not_DP}), and therefore $\Mc$ is not $\paren{1, \frac{\beta}{4n_0}}$-DP (\cref{claim:M-to-A-DP}), contradiction. Thus $d_0 \leq c\cdot n_0^2 \log^2(n_0) \log(n_0/\beta)$ which implies that $n_0 \geq \Omega\paren{\frac{\sqrt{d_0}}{\log^{1.5}(d_0/ \beta)}}$.
\end{proof}

\section{Padding And Permuting (PAP) FPC}\label{sec:FPT:PAP}

In this section, we present the PAP transformation and prove its main property (stated in \cref{lemma:PAP}). In \cref{sec:proving-base-tool,sec:proving-extended-tool} we use our framework (\cref{lemma:framework}) with the PAP transformation to prove \cref{thm:intro:base_tool,thm:intro:extended_tool}.

\begin{definition}[$\PAP_{n, d_0, \ell}$]\label{def:PAP}
	Let $\ell, n, d_0 \in \bbN$, and let $d = d_0 + 2\ell$.
	We define $\PAP_{n, d_0, \ell} \colon \oo^{n \times d_0} \times \cP_d \rightarrow \oo^{n \times d}$ as the function that given $X \in \oo^{n \times d_0}$ and a permutation matrix $P \in \cP_d$ as inputs,
	outputs $X' = X'' \cdot P$ (i.e., permutes the columns of $X''$ according to $P$), where $X''$ is the $\oo^{n\times d}$ matrix after appending $\ell$ $1$-marked and $\ell$ $(-1)$-marked columns to $X$.
\end{definition}

Note that for every $n, d_0, \ell \in \bbN$ and $P \in \cP_d$, the function $\PAP_{n, d_0, \ell}(\cdot, P)$ is neighboring-preserving (\cref{def:neighbor-preserving}). 

The following lemma shows how $\PAP$ can be used to transform strong-agreement into a strong-correlation guarantee.

\begin{lemma}\label{lemma:from_agrees_to_correlated}
	Let $\ell, n, d_0 \in \bbN$ such that $d = d_0 + 2\ell$.
	Let $X \in \oo^{n \times d_0}$,  define the random variables $\bP \la \cP_d$ and  $\bY = \PAP_{n,d_0, \ell}(X, \bP)$,
	and let $\Mc \colon \oo^{n \times d} \rightarrow [-1,1]^{d}$ be a mechanism that for every input $Y \in \Supp(\bY)$, outputs $q \in [-1,1]^d$ that strongly-agrees with $Y$ (\cref{def:intro:strongly-agree}).
	Then $(\Mc(\bY) \cdot \bP^T)^{1,\ldots,d_0}$ is strongly-correlated with $X$ (\cref{def:technique:strongly-correlated}).
\end{lemma}
\begin{proof}
	The proof follows since for every $b \in \oo$ and $j \in \cJ_{X}^b$,
	\begin{align}\label{eq:agg-to-correlation}
		\ppr{(Y,P) \sim (\bY,\bP)}{(\Mc(Y) \cdot P^T)^j = b}
		&= \eex{Y \sim \bY, \: j' \la \cJ_{Y}^b}{\pr{\Mc(Y)^{j'} = b}}\\
		&\geq 0.9 \cdot \eex{Y \sim \bY}{\pr{\Mc(Y)\text{ strongly-agrees with }Y}}\nonumber\\
		&= 0.9,\nonumber
	\end{align}
	where all the probabilities are also taken over the random coins of $\Mc$.
	The equality holds since from the point of view of $\Mc$, which does not know the random permutation $P$, every $b$-marked column of $Y$ has the same probability to be the $b$-marked column $j$ of $X$. The first inequality holds 
	since for every $Y$, 
	\begin{align*}
		\eex{j' \la \cJ_Y^b}{\pr{\Mc(Y)^{j'} = b \mid  \Mc(Y)\text{ strongly-agrees with }Y}} \geq 0.9.
	\end{align*}
	The last inequality in \cref{eq:agg-to-correlation} holds since by the assumption on $\Mc$, for every $Y \in \Supp(\bY)$, the output $\Mc(Y)$ strongly-agrees with $Y$ w.p. $1$.
\end{proof}

We now prove the main property of our PAP technique, which transforms any probability of strong-agreement to the same probability of strong-correlation .

\begin{lemma}\label{lemma:PAP}
	Let $\ell, n, d_0 \in \bbN$ such that $d = d_0 + 2\ell$. Let $\Mc \colon \oo^{n \times d} \rightarrow [-1,1]^{d}$ be a mechanism that uses $m$ random coins,  define the random variable $\bP \la \cP_d$, and for $X \in \oo^{n \times d_0}$ define $\bY_{X} = \PAP(X,\bP)$. Then for any distribution $\cD$ over $\oo^{n \times d_0}$:
	\begin{align*}
		\lefteqn{\ppr{r \la \zo^m, \: X \sim \cD}{(\Mc_{r}(\bY_{X}) \cdot \Pi^T)^{1,\ldots,d_0}\text{ is strongly-correlated with }X}}\\
		&\geq \eex{X \sim \cD}{\pr{\Mc(\bY_{X})\text{ strongly-agrees with } \bY_{X}}}.
	\end{align*}
\end{lemma}
\begin{proof}
	Let $\beta = \eex{X \sim \cD}{\pr{\Mc(\bY_{X})\text{ strongly-agrees with } \bY_{X}}}$, and let $\PAP = \PAP_{n,d_0,\ell}$.
	Also, let $\Mc' \colon \oo^{n \times d} \rightarrow [-1,1]^{d}$ be the mechanism that on input $Y \in \oo^{n \times d}$, samples $ q\sim \Mc(Y)$ and checks if it strongly-agrees with $Y$. If it does, it outputs $q$. Otherwise, it outputs an (arbitrary) vector $q' \in [-1,1]^d$ that strongly-agrees with $Y$.
	In the following, let $\br \la \zo^{m}$, $\bX \sim \cD(n, d_0)$, $\bP \la \cP_d$ and $\bY  = \bY_{\bX} (= \PAP(\bX,\bP))$ be random variables, and let $\bq = \Mc_{\br}(\bY)$ and $\bq' = \Mc'_{\br}(\bY)$.
	Let $E = \set{ (r,X,P) \colon \Mc_r(\PAP(X,P)) = \Mc'_r(\PAP(X,P))}$, and note that $\pr{(\br,\bX,\bP) \in E} = \beta$.
	In addition, for $(r,X) \in \Supp(\br) \times \Supp(\bX)$, let $E_{r,X} = \set{P \colon (r,X,P) \in E}$ and let $\beta_{r,X} = \pr{\bP \in E_{r,X}}$.
	By definition, the following holds:
	\begin{enumerate}
		\item $\forall (r,X) \in \Supp(\br) \times \Supp(\bX)$: $\bq|_{(\br,\bX) = (r,X)}$ behaves as $\bq'|_{(\br,\bX) = (r,X)}$ w.p. $\beta_{r,X}$ (\cref{def:behave_the_same}).\label{item:behave-the-same}
		\item $\eex{r \la \zo^m, \: X \sim \cD}{\beta_{r,X}} = \beta$.\label{item:beta-exp}
		\item $\forall (r,X) \in \Supp(\br) \times \Supp(\bX)$: $\: (\bq'|_{(\br,\bX) = (r,X)} \cdot \bP^T)^{1,\ldots,d_0}$ is strongly-correlated with $X$ (holds by applying \cref{lemma:from_agrees_to_correlated} on the mechanism $\Mc'_{r}$).\label{item:correlated}
	\end{enumerate}
	Thus, we conclude that
	\begin{align*}
		\lefteqn{\ppr{r \la \zo^m, \: X \sim \cD}{(\Mc_{r}(\bY_{X}) \cdot \bP^T)^{1,\ldots,d_0}\text{ is strongly-correlated with }X}}\\
		&= \ppr{r \la \zo^m, \: X \sim \cD}{(\bq|_{(\br,\bX) = (r,X)} \cdot \bP^T)^{1,\ldots,d_0}\text{ is strongly-correlated with }X}\\
		&\geq \eex{r \la \zo^m, \: X \sim \cD}{1 - (1 - \beta_{r,X})}\\
		&= \beta.
	\end{align*}
The inequality holds by \cref{item:behave-the-same,item:correlated,fact:behave_the_same}, and the last equality holds by \cref{item:beta-exp}.

\end{proof}

\subsection{Proving \cref{thm:intro:base_tool} (Basic Tool)}\label{sec:proving-base-tool}

\begin{theorem}[Restatement of \cref{thm:intro:base_tool}]\label{thm:base_tool}
	If $\Mc \colon (\oo^d)^n \rightarrow [-1,1]^d$ is an $(\alpha,\beta)$-weakly-accurate (\cref{def:intro:weakly-accurate}) $(1,\frac{\beta}{4n})$-DP mechanism, then 
	$n \geq \Omega(\sqrt{\alpha d} / \log^{1.5} (\alpha d/\beta))$.
\end{theorem}
\begin{proof}
	We prove the theorem by applying \cref{lemma:framework} with $n_0 = n$ and $d_0 = d- 2\ell$ for $\ell = \ceil{\frac12(1-\alpha)d}$. 
	Let $\PAP = \PAP_{n,d_0,\ell}$ (\cref{def:PAP}), $\cD = \cD(n,d_0)$ (\cref{def:D}), and define $\cV = \cP_d$, $\cW = [-1,1]^d$, $\Gc = \PAP$, $\forall (P, w) \in \cV \times \cW: \Fc(P, w) = (w \cdot P^T)^{1,\ldots,d_0}$, and for $r \in \zo^m$ (random coins for $\Mc$) define $\Ac^{\Mc_r, \Fc, \Gc}$ as in \cref{def:A_MFG} (note that $\Fc$ and $\Gc$ are deterministic).
	By definition, $\Gc(\cdot, P)$ is neighboring-preserving for every $P \in \cP_d$ (\cref{def:neighbor-preserving}). In the following, let $\bP \la \cP_d$ and for $X \in \oo^{n \times d_0}$ define $\bY_{X} = \PAP(X,\bP) (= \Gc(X,\bP))$. 
	Compute
	\begin{align*}
		\lefteqn{\ppr{r \la \zo^m, \: X \sim \cD}{\Ac^{\Mc_r, \Fc, \Gc}(X)\text{ is strongly-correlated with }X}}\\
		&= \ppr{r \la \zo^m, \: X \sim \cD}{(\Mc_r(\bY_{X}) \cdot \bP^T)^{1,\ldots,d_0}\text{ is strongly-correlated with }X}\\
		&\geq \eex{X \sim \cD}{\pr{\Mc(\bY_{X})\text{ strongly-agrees with }\bY_{X}}}\\
		&\geq \beta.
	\end{align*}
	The first inequality hold by \cref{lemma:PAP}. The last inequality holds since $\Mc$ is $(\alpha,\beta)$-weakly-accurate and for every $X \in \oo^{n \times d_0}$ and $Y \in \Supp(Y_X)$ it holds that $\size{\cJ_Y^1}, \size{\cJ_Y^{-1}} \geq \ell \geq \frac12(1-\alpha)d$. 
	
	Thus by \cref{lemma:framework}, $n \geq \Omega(\sqrt{d_0} / \log^{1.5} (d_0/\beta))$, and the proof follows since $d_0 = \Theta(\alpha d)$.
\end{proof}

\subsection{Proving \cref{thm:intro:extended_tool} (Extended Tool)}\label{sec:proving-extended-tool}

We prove \cref{thm:intro:extended_tool} using multiple PAP-FPC copies.

\begin{theorem}[Restatement of \cref{thm:intro:extended_tool}]\label{thm:extended_tool}
	Let $\alpha, \beta \in (0,1]$, $n,k, d \in \bbN$ such that $n$ is a multiple of $k$.
	If $\Mc \colon (\oo^d)^n \rightarrow [-1,1]^d$ is an $(k,\alpha,\beta)$-weakly-accurate (\cref{def:intro:k-weakly-accurate}) $(1,\frac{\beta}{4n})$-DP mechanism, then $n \geq \Omega(k \sqrt{\alpha d} / \log^{1.5} (k \alpha d/\beta))$.
\end{theorem}
\begin{proof}
	We prove the theorem by applying \cref{lemma:framework} with $n_0 = n/k$ and $d_0 = d- 2\ell$ for $\ell = \ceil{\frac12(1-\alpha)d}$.
	Let $\cV = [k]\times(\cP_d)^k$, $\cW = [-1,1]^d$, $\PAP = \PAP_{n_0,d_0,\ell}$, and $\cD = \cD(n_0,d_0)$. Consider the following algorithm $\Hc \colon   \oo^{n_0 \times d} \times \cV \rightarrow \oo^{n \times d}$ that on inputs $Y \in \oo^{n_0 \times d}$ and $v = (s, \vec{P}= (P_1,\ldots,P_k)) \in \cV$, act as follows:
	\begin{enumerate}
		\item Sample $\vec{A} = (A_1,\ldots, A_k) \sim \cD^{k}$.
		\item For $t \in [k]$ set $Y_t = \begin{cases}Y & t=s \\ \PAP(A_t, P_t) & t\neq s\end{cases}$.
	
		Denote $Y_t = (x_{(t-1) n_0 + 1}', \: \ldots, \: x_{t n_0}') \in \oo^{n_0 \times d}$.
		\item Output $X' = (x_1',\ldots,x_n') \in \oo^{n \times d}$.
	\end{enumerate}
	Define $\Gc\colon \oo^{n_0 \times d_0} \times \cV \rightarrow \oo^{n \times d}$ as the algorithm that on inputs $X \in \oo^{n_0 \times d_0}$ and $v = (s,\vec{P}) \in \cV$ for $s \in [k]$ and $\vec{P} = (P_1,\ldots,P_k) \in (\cP_d)^k$:  Computes $Y = \PAP(X, P_s)$, and  outputs $X' \sim \Hc(Y, v)$. Note that for every $v \in \cV$, $\Gc(\cdot ,v)$ is neighboring-preserving (\cref{def:neighbor-preserving}).
	
	Define $\Fc \colon \cV \times [-1,1]^d  \rightarrow [-1,1]^{d_0}$ as the algorithm that on inputs $v = (s,\vec{P}) \in \cV$ and  $w \in [-1,1]^d$, outputs $(w \cdot P_s^T)^{1,\ldots,d_0}$.
	Note that $\Fc$ is deterministic, and the random choice of $\Gc$ is $\vec{A} \sim \cD^{k}$ (chosen when executing $\Hc$).
	
	Our goal is to show that 
	\begin{align}\label{eq:extended-case:goal}
		\ppr{r \la \zo^m, \: \vec{A}\sim \cD^{k}, \: X \sim \cD}{\Ac^{\Mc_r, \Fc, \Gc_{\vec{A}}}(X)\text{ is strongly-correlated with }X} \geq \beta/k.
	\end{align}
	 Given that \cref{eq:extended-case:goal} holds and since $\Mc$ is $\paren{1,\frac{\beta}{4 k n_0}}$-DP, we deduce by \cref{lemma:framework} that $n_0 \geq \Omega(\sqrt{d_0} / \log^{1.5} (k d_0/\beta))$ and the proof follows since $n_0 = n/k$ and $d_0 = \Theta(\alpha d)$. 

	We now focus on proving \cref{eq:extended-case:goal}.
	Consider the mechanism $\Mc' \colon \oo^{n_0 \times d} \rightarrow [-1,1]^d$ that on input $Y \in \oo^{n_0 \times d}$, samples $v = (s,\vec{P}) \la\cV$, computes $X' = \Hc(Y, v)$ and outputs $w \sim \Mc(X')$.
	
	In the following, let $\bP \la \cP_k$ be a random variable, and for $X \in \oo^{n_0 \times d_0}$ define $\bY_{X} = \PAP(X, \bP)$.
	Compute
	%\underset{}{}
	\begin{align*}
		\lefteqn{\eex{X \sim \cD}{\pr{\Mc'(\bY_{X})\text{ strongly-agrees with } \bY_{X}}}}\\
		&= \eex{X \sim \cD, \: v = (s,\vec{P}) \la \cV, \: \vec{A} \sim \cD^k}{\pr{\Mc(\Gc_{\vec{A}}(X, v))\text{ strongly-agrees with } \PAP(X,P_s)}}\\
		&=  \eex{v = (s,\vec{P}) \la \cV, \: \vec{A} \sim \cD^k}{\pr{\Mc(\Gc_{\vec{A}}(A_s, v))\text{ strongly-agrees with } \PAP(A_s,P_s)}}\\
		&= \eex{v = (s,\vec{P}) \la \cV, \: \forall t \in [k]: \:Y_t = (x'_{(t-1) n_0 + 1}, \: \ldots, \: x'_{t n_0}) \sim \PAP(\cD, P_t)}{\pr{\Mc(x_1',\ldots,x_n')\text{ strongly-agrees with } Y_s}}\\
		&\geq \beta/k.
	\end{align*}
	The first equality holds since by the definitions of $\Mc'$ and $\Gc$, for every $X$, the pair $\paren{\Mc'(\bY_{X}), \bY_{X}}$ has the same distribution as the pair $\paren{\Mc(\Gc_{s,\vec{A}}(X, v)), \PAP(X,P_s)}_{v = (s,\vec{P}) \la \cV, \: \vec{A} \sim \cD^k}$. The second equality holds since $\Gc_{\vec{A}}(\cdot, (s,\vec{P}))$ does not use the value of $A_s$ which is drawn as $X$. The third equality holds by the definition of $\Gc$. The inequality holds since $\Mc$ is an $(k,\alpha,\beta)$-weakly-accurate mechanism, and it gets $X' = (x_1',\ldots,x_n')$ where each $Y_t = (x'_{(t-1) n_0 + 1}, \: \ldots, \: x'_{t n_0})$ has $\size{\cJ_{Y_t}^1}, \size{\cJ_{Y_t}^{-1}} \geq \ell \geq \frac12(1-\alpha)d$. Thus w.p. $\beta$, $\Mc$ is guaranteed to output $w \in [-1,1]^d$ that strongly-agrees with one of the $Y_1,\ldots,Y_k$. But since they are drawn independently from the same distribution, the probability to agree with $Y_s$ for a random $s \la [k]$ decreases by a factor $k$, and overall is at least $\beta/k$.

	We now apply \cref{lemma:PAP} to get that
	\begin{align*}
		\ppr{\underset{r \la \zo^{m}, \: X \sim \cD}{\vec{A} \sim \cD^{k}, \: v = (s, \vec{P}) \la \cV}}{\paren{\Mc'_{v,\vec{A},r}(\bY_{X}) \cdot \bP^T}^{1,\ldots,d_0} \text{ is strongly-correlated with } X} \geq \beta/k.
	\end{align*}
	Since $(\Mc'_{v,\vec{A},r}(\bY_{X}) \cdot \bP^T)_{r \la \zo^{m}, \: \vec{A} \sim \cD^{k}, \: v \la \cV}$ has the same distribution as\\$\paren{\Mc_r(\Gc_{\vec{A}}(X, \bv)) \cdot \bP_{\bs}^T}_{r \la \zo^m, \: \vec{A} \sim \cD^k}$ for $\bv = (\bs, \vec{\bP} = (\bP_1,\ldots,\bP_k)) \la \cV$, we deduce that
	\begin{align*}%\label{eq:k-guarantee}
		\ppr{X \sim \cD, \: r \la \zo^m, \: \vec{A} \sim \cD^k}{\paren{\Mc_r(\Gc_{\vec{A}}(X, \bv)) \cdot \bP_{\bs}^T}^{1,\ldots,d_0} \text{ is strongly-correlated with } X} \geq \beta/k.
	\end{align*}
	We thus conclude that
	\begin{align*}
		\lefteqn{\ppr{r \la \zo^m, \vec{A} \sim \cD^{k}, \: X \sim \cD}{\Ac^{\Mc_r, \Fc, \Gc_{\vec{A}}}(X)\text{ is strongly-correlated with }X}}\\
		&= \ppr{r \la \zo^m, \vec{A} \sim \cD^{k}, \: X \sim \cD}{\Fc(\bv, \Mc_r(\Gc_{\vec{A}}(X, \bv)))\text{ is strongly-correlated with }X}\\
		&= \ppr{r \la \zo^m, \vec{A} \sim \cD^{k}, \: X \sim \cD}{(\Mc_r(\Gc_{\vec{A}}(X, \bv)) \cdot \bP_{\bs}^T)^{1,\ldots,d_0}\text{ is strongly-correlated with }X}\\
		&\geq \beta/k.
	\end{align*}
\end{proof}

\section{Applications}\label{sec:applications}

Throughout this section, recall for $z \in \bbR$ we define $\sign(z)\eqdef \begin{cases} 1 & z \geq 0 \\ -1 & z<0\end{cases}$
and for $u = (u^1,\ldots,u^d) \in \bbR^d$ we define $\sign(u) \eqdef (\sign(u^1),\ldots,\sign(u^d))$.

\subsection{Averaging}

In this section, we prove \cref{thm:intro:lower_bound:avg}.

\begin{definition}[$(\lambda,\beta)$-Estimator for Averaging, Redefinition of \cref{def:intro:avg-est}]
    A mechanism $\Mc \colon \bbR^+ \times (\bbR^{d})^n \rightarrow \bbR^d$ is \emph{$(\lambda,\beta)$-estimator for averaging} if given $\gamma \geq 0$ and $X = (x_1,\ldots,x_n) \in \bbR^{n \times d}$ with $\max_{i,j \in [n]} \norm{x_i - x_j}_2  \leq \gamma$, it holds that 
    \begin{align*}
        \pr{\norm{\Mc(\gamma,X) - \frac1n \sum_{i=1}^n x_i}_2 \leq \lambda \gamma} \geq \beta.
    \end{align*}
\end{definition}

\begin{theorem}[Our averaging lower bound, Restatement of \cref{thm:intro:lower_bound:avg}]\label{thm:lower_bound:avg}
    If $\Mc \colon \bbR^+ \times (\bbR^{d})^n \rightarrow \bbR^d$ is a \emph{$(\lambda, \beta)$-estimator for averaging} for $\lambda \geq 1$ and $\Mc(\gamma,\cdot)$ is $\paren{1,\frac{\beta}{4n}}$-DP for every $\gamma \geq 0$, then $n \geq \Omega\paren{\frac{\sqrt{d}/\lambda}{\log^{1.5} \paren{\frac{d}{\beta \lambda}}}}$.
\end{theorem}

The proof of the theorem immediately follows by \cref{thm:intro:base_tool} and the following \cref{claim:avg} (along with post-processing of differential privacy).

\begin{claim}\label{claim:avg}
If $\Mc \colon \bbR^+ \times (\bbR^d)^n \rightarrow \bbR^d$ is a \emph{$(\lambda, \beta)$-estimator for averaging}, then the mechanism $\tMc\colon (\oo^d)^n \rightarrow \oo^d$ defined by $\tMc(X) \eqdef \signn{\Mc(\gamma=\sqrt{2 \alpha d},\: X)}$ for $\alpha=\frac{1}{40\lambda^2 + 1}$, is $(\alpha, \beta)$-weakly-accurate (\cref{def:intro:weakly-accurate}).

\end{claim}
\begin{proof}
    Fix $X=(x_1,\ldots,x_n) \in \oo^{n\times d}$ with $\size{\cJ^1_X},  \size{\cJ^{-1}_X} \geq \frac12(1-\alpha)d$ and let $\mu = (\mu^1,\ldots,\mu^d) = \frac1n \sum_{i=1}^n x_i$. Since $x_1,\ldots,x_n \in \oo^d$ agree on at least $(1-\alpha)d$ coordinates, their diameter is bounded by $\gamma = \sqrt{2 \alpha d}$. By the utility guarantee of $\Mc$ it holds that
    \begin{align}\label{eq:avg-g}
    	\ppr{q \sim \Mc(\gamma,X)}{\norm{q -\mu}_2 \leq \lambda \gamma} \geq \beta.
    \end{align}
	We prove the claim by showing that for any $q \in \bbR^d$ with $\norm{q -\mu}_2 \leq \lambda \gamma$ it holds that $\sign(q)$ strongly-agrees with $X$ (\cref{def:intro:strongly-agree}).
	Fix such $q$. Note that $\size{\set{j \in [d] \colon \sign(q^j) \neq \sign(\mu^j)}} \leq \norm{q -\mu}_2^2 \leq \lambda^2 \gamma^2$. Since $\mu^j = b$ for every $b \in \oo$ and $j \in \cJ_X^b$, we deduce that $\size{\set{j \in \cJ_X^b \colon \sign(q^j) \neq b}} \leq \lambda^2 \gamma^2$. Now note that for both $b \in \oo$:
	\begin{align*}
		\lambda^2 \gamma^2 = 2\alpha \lambda^2 d \leq 0.1 \cdot \frac12(1-\alpha)d \leq 0.1 \cdot \size{\cJ_X^b},
	\end{align*}
	where the first inequality holds by our choice of $\alpha$.
	Thus $\sign(q)$ strongly-agrees with $X$, as required.
\end{proof}

\subsection{Clustering}\label{sec:k-means}

In this section, we prove an extension of Theorem~\ref{thm:intro:kmeans} to $(k,z)$-clustering where we focus on the Euclidean metric space $(\bbR^d, d(x,y) \eqdef \norm{x-y}_2)$.
We start by extending the notations from \cref{sec:intro:clustering}.

Let $\cB_d \eqdef \set{x \in \bbR^d \colon \norm{x}_2 \leq 1}$. For a database $S \in (\cB_d)^n$, $k$ centers $C=(c_1,\ldots,c_k) \in (\cB_d)^k$ and a parameter $z \geq 1$, let 
$$\COST_{z}(C ; S) \eqdef \sum_{x \in S} \min_{i \in [k]}\norm{x - c_i}_2^z \quad \text{ and } \quad \OPT_{k,z}(S) \eqdef \min_{C \in (\bbR^d)^k} \COST_{z}(C ; S).$$

\begin{definition}[$(\lambda,\xi, \beta)$-Approximation Algorithm for $(k,z)$-Clustering, Redefinition of Definition~\ref{def:intro:k_means_approx}]
    We say that $\Mc\colon (\cB_d)^n \rightarrow (\cB_d)^k$ is an \emph{$(\lambda,\xi, \beta)$-approximation algorithm for $(k,z)$-clustering}, if for every $S \in (\cB_d)^n$ it holds that
    \begin{align*}
        \ppr{C \sim \Mc(S)}{\COST_{z}(C;S) \leq \lambda \cdot \OPT_{k,z}(S) + \xi} \geq \beta
    \end{align*}
\end{definition}

\begin{theorem}[Our Lower Bound, Extension of Theorem~\ref{thm:intro:kmeans}]\label{thm:lower_bound:kmeans}
	Let $n,d,k \in \bbN$, $\lambda, z \geq 1$ and $\xi \geq 0$  such that $k \geq 2$ and $n \geq k + 2 \cdot 40^{z/2} \xi$. If $\Mc \colon (\cB_d)^n \rightarrow (\cB_d)^k$ is an $(1,\frac{\beta}{4nk})$-DP $(\lambda,\xi, \beta)$-approximation algorithm  for $(k,z)$-clustering, then either $k \geq 2^{\Omega(d/\lambda^{2/z})} \beta \lambda^{2/z}/d \:$ or $ \:\xi \geq \Omega\paren{\frac{2^{-O(z)} k \sqrt{d/\lambda^{2/z}}}{\log^{1.5}\paren{\frac{kd}{\beta \lambda^{2/z}}}}}$.
\end{theorem}

The following claim captures the main technical part in proving Theorem~\ref{thm:lower_bound:kmeans}.

\begin{claim}\label{claim:clustering}
	Let $n,d,k \in \bbN$, $\lambda,z \geq 1$ and $\xi \geq 1$ such that $n \geq m$ for $m  = k \cdot \floor{(1 + 40^{z/2}\cdot 2\xi /k)}$.
	If $\Mc \colon (\cB_d)^n \rightarrow (\cB_d)^{k+1}$ is an $(\lambda,\: \xi, \: \beta)$-approximation algorithm for $(k+1,z)$-clustering, then the following mechanism $\tMc\colon (\oo^d)^m \rightarrow \oo^d$ is $(k, \: \alpha =\frac{1}{160\cdot (2\lambda)^{2/z}}, \: \frac{\beta}{k+1})$-weakly-accurate (Definition~\ref{def:intro:k-weakly-accurate}).
\end{claim}

\begin{algorithm}[A $(k, O(1/\lambda), \beta)$-weakly-accurate variant $\tMc$ of the $(k+1)$-means $(\lambda,\xi, \beta)$-approximation algorithm $\Mc$]\label{alg:tM}
	\item Input: $x_1,\ldots,x_m \in \oo^d$ for $m = \floor{(1 + 40^{z/2}\cdot 2\xi /k)}\cdot k$.
	\item Operation:~
	\begin{enumerate}
		\item Compute $(c_1,\ldots,c_{k+1}) = \Mc(\frac1{\sqrt{d}} x_1,\ldots,\frac1{\sqrt{d}} x_m, \underbrace{\vec{0},\ldots,\vec{0}}_{n-m\text{ times}})$ (where $\vec{0} = (\underbrace{0,\ldots,0}_{d\text{ times}})$).\label{tM:step:tc}
		
		\item Sample $j \la [k+1]$ and output $\signn{c_j}$.\label{step:j}
		
	\end{enumerate}
\end{algorithm}

We first prove Theorem~\ref{thm:lower_bound:kmeans} using Claim~\ref{claim:clustering}.

\begin{proof}[Proof of Theorem~\ref{thm:lower_bound:kmeans}]
	Let $\Mc \colon (\cB_d)^n \rightarrow (\cB_d)^k$ be the mechanism from Theorem~\ref{thm:lower_bound:kmeans}. By  Claim~\ref{claim:clustering}, the mechanism $\tMc\colon (\oo^d)^m \rightarrow \oo^d$ is $(k-1, \alpha =\frac{1}{160\cdot (2\lambda)^{2/z}}, \frac{\beta}{k})$-weakly-accurate. Furthermore, by post-processing, $\tMc$ is also $\paren{1, \frac{\beta}{4nk}}$-DP, which in particular implies that it is $\paren{1, \frac{\beta}{4n(k-1)}}$-DP. Hence we deduce by Theorem~\ref{thm:intro:extended_tool} that $m \geq \Omega\paren{\frac{k\sqrt{d/\lambda^{2/z}}}{\log^{1.5}\paren{\frac{dk}{\beta \lambda^{2/z}}}}}$. Since $m \geq \max\set{k, 40^{z/2}\cdot 2\xi}$, the above implies that either $k \geq \Omega\paren{\frac{ k\sqrt{d/\lambda^{2/z}}}{ \log^{1.5}\paren{\frac{dk}{\beta \lambda^{2/z}}}}}$ (which implies that $k \geq 2^{\Omega(d/\lambda^{2/z})}\beta \lambda^{2/z}/d \:$\remove{$k \geq 2^{\Omega(d/\lambda^{2/z})}\beta \lambda^{2/z}/d$}), or that $\xi \geq \Omega\paren{\frac{2^{-O(z)} k\sqrt{d/\lambda^{2/z}}}{\log^{1.5}\paren{\frac{dk}{\beta \lambda^{2/z}}}}}$.
\end{proof}

We now prove Claim~\ref{claim:clustering}.

\begin{proof}[Proof of Claim~\ref{claim:clustering}]
Fix $X = (x_1,\ldots,x_m) \in \oo^{m \times d}$ such that for every $t \in [k]$ and every $b \in \oo$ it holds that $\size{\cJ^b_{X_t}} \geq \frac12(1-\alpha)d$ for $X_t = (x_{(t-1)m/k + 1}, \ldots x_{tm/k})$,
and let $S = (\tx_1,\ldots,\tx_m, \underbrace{\vec{0},\ldots,\vec{0}}_{n-m\text{ times}})$ for $\tx_i = \frac1{\sqrt{d}} x_i$. We first argue that $\OPT_{k+1}(S)$ is small.
For $t \in [k]$ define $c_t^* = \tx_{tm/k}$ and $c_{k+1}^* = \vec{0}$ (which covers all the $\vec{0}$'s with zero cost). 
By assumption, for every $t \in [k]$ and $\tx_i \in S$ with $x_i \in X_t$ it holds that
\begin{align*}
	\norm{c_t^* - \tx_i}_2^z 
	= \paren{\norm{c_t^* - \tx_i}_2^2}^{z/2}
	= \paren{\frac1{d} \cdot \norm{x_{tm/k} - x_i}_2^2}^{z/2}
	\leq (4\alpha)^{z/2},
\end{align*}
where the inequality holds since the number of indices that the vectors do not agree on is at most $d - \size{\cJ^1_{X_t}}- \size{\cJ^{-1}_{X_t}} \leq \alpha d$.
Hence, we deduce that
\begin{align*}%\label{eq:opt:kmeans}
	\OPT_{k+1, z}(S) \leq \COST_{z}(c_1^*,\ldots,c_{k+1}^* ; S) \leq (4\alpha)^{z/2} m.
\end{align*}
Now, since
\begin{align*}
	\lambda \cdot \OPT_{k+1, z}(S) + \xi
	&\leq \lambda \cdot (4\alpha)^{z/2} m + \frac{m}{2 \cdot 40^{z/2}}\\
	&= \lambda\cdot \paren{\frac1{40 (2\lambda)^{2/z}}}^{z/2} + \frac{m}{2 \cdot 40^{z/2}}\\
	&\leq m/40^{z/2},
\end{align*}
the utility guarantee of $\Mc$ implies that
\begin{align*}
	\ppr{C = (c_1,\ldots,c_{k+1}) \sim \Mc(S)}{\COST_{z}(C; S) \leq m/40^{z/2}} \geq \beta.
\end{align*}
We prove the claim by showing that for every $C= (c_1,\ldots,c_{k+1})$ with $\COST_{z}(C; S) \leq m/40^{z/2}$, there exists $s \in [k+1]$ and $t \in [k]$ such that $\sign(c_s)$ strongly-agrees with $X_t$. This will conclude the proof since the probability that the random $j$ chosen in \stepref{step:j} of $\tMc$ will hit the right $s$ is $1/(k+1)$.
Indeed if this is not the case, then for any center $c_s$ and any non-zero point $\tx_i \in S$ where $x_i \in X_t$ we have at least one $b \in \oo$ such that $\size{\set{j \in \cJ_{X_t}^b \colon \sign(\tx_i^j) \neq \sign(c_s^j)}} > 0.1\cdot \size{\cJ_{X_t}^b} > \frac1{20}(1-\alpha)d$, which yields that
\begin{align*}
	\norm{\tx_i - c_s}_2^z 
	&= (\norm{\tx_i - c_s}_2^2)^{z/2}\\
	&\geq \paren{\frac1{d}\cdot \size{\set{j \in [d] \colon \sign(\tx_i^j) \neq \sign(c_s^j)}}}^{z/2}\\
	&\geq \paren{\frac1{d}\cdot \size{\set{j \in \cJ_{X_t}^b \colon \sign(\tx_i^j) \neq \sign(c_s^j)}}}^{z/2}\\
	&> \paren{\frac1{20}(1-\alpha)}^{z/2}
\end{align*}
Thus $\COST_{z}(C ;S) > \paren{\frac1{20}(1-\alpha)}^{z/2}\cdot m > m/40^{z/2}$, a contradiction to the assumption $\COST_{z}(C; S) \leq m/40^{z/2}$. This concludes the proof of the claim.

\end{proof}

\subsection{Top Singular Vector}\label{sec:sing-vec}
In this section, we prove \cref{thm:intro:top_sing}. We start by recalling the setting.

For a matrix $X \in \bbR^{n \times d}$, the singular value decomposition of $X$ is defined by $X = U \Sigma V^T$, where $U \in \bbR^{n \times n}$ and  $V \in \bbR^{d \times d}$ are unitary matrices. The matrix $\Sigma \in \bbR^{n \times d}$ is a diagonal matrix with non-negative entries $\sigma_1 \geq \ldots \geq \sigma_{\min\set{n,d}} \geq 0$ along the diagonal, called the singular values of $X$. It is well known that $\norm{X}_F^2 \eqdef \sum_{i \in [n], j \in [d]} (x_i^j)^2 = \sum_{i} \sigma_i^2$. The top (right) singular vector of $X$ is defined by the first column of $V$ (call it $v_1 \in \cS_d$) which satisfy $\norm{X \cdot v_1}_2 = \max_{v \in \cS_d} \norm{X \cdot v}_2$.

In the problem we consider, $n$ rows $x_1,\ldots,x_n \in \cS_d \eqdef \set{v \in \bbR^d \colon \norm{v}_2 = 1}$ are given as input, and the goal is to estimate the top (right) singular vector of the $n \times d$ matrix $X = (x_i^j)_{i \in [n], j \in [d]}$.

\begin{definition}[$(\lambda, \beta)$-Estimator of Top Singular Vector, Redefinition of \cref{def:intro:top_sing_est}]
    We say that $\Mc \colon [0,1] \times (\cS_d)^n \rightarrow \cS_d$ is an $(\lambda, \beta)$-estimator of top singular vector, if given an $n \times d$ matrix $X = (x_1,\ldots,x_n) \in (\cS_d)^{n}$ and a number $\gamma \in [0,1]$ such that $\sigma_2(X) \leq \gamma \cdot \sigma_1(X)$ as inputs, outputs a column vector $y \in \cS_d$ such that
    \begin{align*}
        \ppr{y \sim \Mc(\gamma, X)}{\norm{X\cdot y}_2^2 \geq \norm{X\cdot v}_2^2 - \lambda\cdot \gamma n} \geq \beta,
    \end{align*}
    where $v$ is the top (right) singular vector of $X$.
\end{definition}

We now restate and prove \cref{thm:intro:top_sing}.

\begin{theorem}[Our lower bound, restatement of \cref{thm:intro:top_sing}]\label{thm:lower_bound:top_sing}
    If $\Mc \colon [0,1] \times (\cS_d)^n \rightarrow \cS_d$ is a \emph{$(\lambda,\beta)$-estimator of top singular vector} for $\lambda \geq 1$ and $\Mc(\gamma,\cdot)$ is $(1,\frac{\beta}{4n})$-DP for every $\gamma \in [0,1]$, then $n = \Omega\paren{\frac{\sqrt{d}/\lambda}{\log \paren{\frac{d}{\beta \lambda}}}}$.
\end{theorem}

The proof of \cref{thm:lower_bound:top_sing} immediately follows by \cref{thm:base_tool} and the following claim.

\begin{claim}
    If $\Mc \colon [0,1] \times (\cS_d)^n \rightarrow \cS_d$ is an \emph{$(\lambda, \beta)$-estimator for top singular vector} for $\lambda \geq 1$, then either the mechanism $\tMc\colon \oo^{n \times d} \rightarrow \oo^d$ defined by $\tMc(X) \eqdef \signn{ \Mc(\gamma= \sqrt{\frac{2\alpha}{1-2\alpha}}, \: \frac1{\sqrt{d}} X)}$, for $\alpha = \frac1{4000 \lambda^2}$, or the mechanism $\tMc'(X) \eqdef - \tMc(X)$, are $(\alpha,\beta/2)$-weakly-accurate (\cref{def:intro:weakly-accurate}).
\end{claim}
\begin{proof}
	Fix $X=(x_1,\ldots,x_n) \in \oo^{n\times d}$ with $\size{\cJ^1_X},  \size{\cJ^{-1}_X} \geq \frac12(1-\alpha)d$, and let $\tX = \frac1{\sqrt{d}} X \in (\cS_d)^n$. Let $u = (u_1,\ldots,u_d) \in \cS_d$ be a column vector with
	$u_j = \begin{cases}1/\sqrt{d} & j \in  \cJ^{1}_X \\ -1/\sqrt{d} & \text{o.w.}\end{cases}$. By definition it holds that $\norm{\tX \cdot u}_2^2 \geq \frac{n}d \paren{ 2\size{\cJ_X^1} + 2\size{\cJ_X^{-1}} - d} \geq (1-2\alpha)n$.
	This yields that $\sigma_1(\tX)^2 \geq (1-2\alpha) n$. Since $\sigma_1(\tX)^2 +  \sigma_2(\tX)^2 \leq \norm{\tX}_F^2 = n$, we deduce that $\sigma_2(\tX)^2/\sigma_1(\tX)^2 \leq \frac{2\alpha}{1-2\alpha} = \gamma^2$.
	
	Therefore, by the utility guarantee of $\Mc$, it holds that
	\begin{align*}
		\ppr{y \sim \Mc(\gamma, \: \tX)}{\norm{\tX\cdot y}_2^2 \geq \norm{\tX\cdot u}_2^2 - \lambda\cdot \gamma n} \geq \beta.
	\end{align*}
	We prove the claim by showing that for any $y = (y_1,\ldots,y_d) \in \cS_d$ with $\norm{\tX\cdot y}_2^2 \geq \norm{\tX\cdot u}_2^2 - \lambda\cdot \gamma n$, either $\sign(y)$ or $-\sign(y)$ strongly-agrees with $X$ (\cref{def:intro:strongly-agree}).
	
	In the following, fix such $y$, and let $s = \sum_{j \in \cJ_X^1} y_j - \sum_{j \in \cJ_X^{-1}} y_j$. We next prove that
	\begin{enumerate}
		\item $\norm{\tX\cdot u}_2^2 - \norm{\tX\cdot y}_2^2 \geq \paren{1 - \frac{2 s^2}{d} - 4\alpha} n$, and\label{item:tsv:norms}
		\item If $\sign(y)$ and $-\sign(y)$ do not strongly-agree with $X$, then $s^2 \leq \paren{1-\frac1{40}}d$.\label{item:tsv:s}
	\end{enumerate}
	The proof of the claim follows by \cref{item:tsv:norms,item:tsv:s} since if $\sign(y)$ and $-\sign(y)$ do not strongly-agree with $X$, then by our choice of $\alpha$ it holds that
	\begin{align*}
		\norm{\tX\cdot u}_2^2 - \norm{\tX\cdot y}_2^2 \geq \paren{\frac1{40} - 4\alpha}n > \lambda \underbrace{\sqrt{\frac{2\alpha}{1-2\alpha}}}_{\gamma} n,
	\end{align*}
	which is a contradiction to the assumption about $v$.
	
	We first prove \cref{item:tsv:norms}. Let $z = (z_1,\ldots,z_n) = \tX \cdot y$ and let $\cB = [d]\setminus (\cJ_X^1 \cup \cJ_X^{-1})$.
	Note that for every $i \in [n]$:
	\begin{align*}
		\size{z_i} = \size{\frac1{\sqrt{d}} \sum_{j=1}^d x_i^j \cdot y_j} = \frac1{\sqrt{d}} \size{\sum_{j \in \cJ_X^1} y_j - \sum_{j \in \cJ_X^{-1}} y_j + \sum_{j \in \cB} x_i^j \cdot y_j} \leq \frac{\size{s}}{\sqrt{d}} + \sqrt{\alpha},
	\end{align*}
	where the last inequality holds since $\sum_{j \in \cB} x_i^j \cdot y_j  \leq \norm{x_i^{\cB}}_2 \cdot \norm{y_{\cB}}_2 \leq \sqrt{\size{\cB}} \cdot 1 \leq \sqrt{\alpha d}$.
	
	Using the inequality $(a+b)^2 \leq 2a^2 + 2b^2$ we deduce that $\norm{z}_2^2 = \sum_{i=1}^n z_i^2 \leq \paren{\frac{2s^2}{d} + 2\alpha}n$. Hence
	\begin{align*}
		\norm{\tX\cdot u}_2^2 - \norm{z}_2^2 \geq (1 - 2\alpha)n - \paren{\frac{2s^2}{d} + 2\alpha}n = \paren{1 - \frac{2 s^2}{d} - 4\alpha} n.
	\end{align*}

	We next prove \cref{item:tsv:s}.
	Assume $\sign(y)$ and $-\sign(y)$ do not strongly-agree with $X$. Then there exists $b,b' \in \oo$ such that
	\begin{align}\label{eq:sing:cont1}
		\eex{j \la \cJ^{b}_X}{y_j \neq b} > 0.1
	\end{align}
	and 
	\begin{align}\label{eq:sing:cont2}
		\eex{j \la \cJ^{b'}_X}{y_j = b'} > 0.1.
	\end{align}

	We split into two cases:
	
	\paragraph{Case 1: \cref{eq:sing:cont1,eq:sing:cont2} holds for $b = b'$.}
	
	In this case, we assume for simplicity that $b = b' = 1$ (the $b = b' = -1$ case holds by symmetry). Let $\cJ^+  = \set{j \in \cJ_X^1 \colon y_j \geq 0}$ and $\cJ^-  = \set{j \in \cJ_X^1 \colon y_j < 0}$.
	Since $\signn{y_j} \neq 1 \implies y_j < 0$ and $\signn{y_j} = 1 \implies y_j \geq 0$, \cref{eq:sing:cont1,eq:sing:cont2} yields that
	\begin{align}\label{eq:sing_vec:first_max}
		\size{\cJ^+ } \leq 0.9 \size{\cJ_X^1} \quad \text{ and } \quad \size{\cJ^- } \leq 0.9 \size{\cJ_X^1}.
	\end{align}
	Now recall that $s = \sum_{j \in \cJ_X^1} y_j - \sum_{j \in \cJ_X^{-1}} y_j$.
	Therefore, 
	\begin{align}\label{eq:top_sing:J+}
		s 
		&\leq \sum_{j \in \cJ^+ } y_j - \sum_{j \in \cJ_X^{-1}} y_j
		\leq \sum_{j \in \cJ^+ \cup \cJ_X^{-1}} \size{y_j}\\
		&= \norm{y_{\cJ^+ \cup \cJ_X^{-1}}}_1
		\leq \sqrt{\size{\cJ^+ \cup \cJ_X^{-1}}} \cdot \norm{y_{\cJ^+ \cup \cJ_X^{-1}}}_2\nonumber\\
		&\leq \sqrt{\size{\cJ^+ \cup \cJ_X^{-1}}}\nonumber
	\end{align}
	where the third inequality holds since $\norm{w}_1 \leq \sqrt{m}\cdot \norm{w}_2$ for every $w \in \bbR^m$, and the fourth one holds since $\norm{y_{\cJ^+ \cup \cJ_X^{-1}}}_2 \leq \norm{y}_2 = 1$.

	By the right inequality in (\ref{eq:sing_vec:first_max}), a similar calculation to \cref{eq:top_sing:J+} yields that
	\begin{align}\label{eq:top_sing:J-}
		s \geq \sum_{j \in \cJ^- } y_j - \sum_{j \in \cJ^{-1}_X} y_j \geq - \sum_{j \in \cJ^-  \cup \cJ_X^{-1}} \size{y_j} \geq \ldots
		\geq -\sqrt{\size{\cJ^-  \cup \cJ^{-1}_X}}.
	\end{align}
	The proof in this case now follows by \cref{eq:top_sing:J+,eq:top_sing:J-} since 
	\begin{align*}
		\size{\cJ^+  \cup \cJ^{-1}_X}, \:\size{\cJ^+  \cup \cJ^{-1}_X} \leq d - 0.1 \size{\cJ_X^1} \leq d - \frac1{20}(1-\alpha) d \leq \paren{1 - \frac1{40}}d,
	\end{align*}
	where the first inequality holds by \cref{eq:sing_vec:first_max}.

	\paragraph{Case 2: \cref{eq:sing:cont1,eq:sing:cont2} holds for $b \neq b'$.}
	
	In this case, we assume for simplicity that $b = 1$ and $b' = -1$ (the $b=-1$ and $b'=1$ case holds by symmetry).
	Let $\cJ^1 = \set{j \in \cJ_X^1 \colon y_j \geq 0}$ and $\cJ^{-1} = \set{j \in \cJ_X^{-1} \colon y_j \geq 0}$. Since $\signn{y_j} \neq 1 \implies y_j < 0$, \cref{eq:sing:cont1,eq:sing:cont2} yield that 
	
	\begin{align}\label{eq:sing_vec:second_max}
		\size{\cJ^1} \leq 0.9 \size{\cJ_X^1} \quad \text{ and } \quad \size{\cJ^{-1}} \leq 0.9 \size{\cJ_X^{-1}}
	\end{align}

	The left inequality in \cref{eq:sing_vec:second_max} implies that
	\begin{align}\label{eq:top_sing:J1}
		s
		\leq \sum_{j \in \cJ^1} y_j - \sum_{j \in \cJ_X^{-1}} y_j
		\leq \sum_{j \in \cJ^1 \cup \cJ_X^{-1}} \size{y_j}
		\leq \sqrt{\size{\cJ^1\cup \cJ_X^{-1}}}
	\end{align}
	where the last inequality holds similarly to \cref{eq:top_sing:J+}.
	Similarly, the right inequality in \cref{eq:sing_vec:second_max} implies that
	\begin{align}\label{eq:top_sing:J-1}
		s \geq  \sum_{j \in \cJ_X^1} y_j - \sum_{j \in \cJ^{-1}} y_j \geq -\sum_{\cJ_X^1 \cup \cJ^{-1}} \size{y_j} \geq -\sqrt{\size{\cJ_X^1 \cup \cJ^{-1}}}.
	\end{align}
	The proof now follows by \cref{eq:top_sing:J1,eq:top_sing:J-1} since $\size{\cJ^1\cup \cJ_X^{-1}} \leq d - 0.1 \size{\cJ_X^1} \leq \paren{1 - \frac1{40}}d$ and 
	$\size{\cJ_X^1 \cup \cJ^{-1}} \leq d - 0.1 \size{\cJ_X^{-1}} \leq \paren{1 - \frac1{40}}d$.

\end{proof}

\ifdefined\IsAnonymous
\else
\section*{Acknowledgments}

Eliad Tsfadia would like to thank Edith Cohen, Haim Kaplan, Yishay Mansour and Uri Stemmer for encouraging him to tackle the problem of lower bounding DP averaging and for useful discussions.
\fi

%\addcontentsline{References}
\printbibliography

\appendix

\section{Fingerprinting Codes}\label{sec:appendix:FPC}

An $(n,d)$-fingerprinting code, where $[n]$ is the set of users and $d$ is the code length, is a pair $(\Gen,\Trace)$. The first element $\Gen$ is a randomized algorithm that generates and outputs a codebook $X \in \oo^{n \times d}$, where for every $i \in [n]$, the \ith row $x_i \in \oo^d$ of $X$ is the code of user $i$. 
The second element $\Trace$ is a randomized algorithm that receives as an input a codebook $X \in \oo^{n \times d}$ and a word $q \in \oo^d$ that was (possibly) generated by a malicious subset of users $\cS \subseteq [n]$ (from their codes), and outputs a user $i \in [n]$ or $\perp$. 

The security requirements of fingerprinting codes says that (1) given every code $x_i$, one can be verify that user $i$ holds it (that is, it is known that user $i$ is the one to receive the code $x_i$), and that (2) given a codebook $X$ and a word $q$ generated by a malicious subset of users $\cS$ using their codes, the algorithm $\Trace$ returns a user $i \in \cS$ with high probability (so, it returns a user $i \notin \cS$ or $\perp$ with low probability).  

The basic assumption of fingerprinting code is that for every $j \in [d]$, the \jth element of $q$ must be equal to the \jth element of $x_i$ for some $i\in \cS$.
In particular, if the \jth element of all the codes of $\cS$ is equal to some $b \in \oo$, then the \jth element of the word $q$ must also be $b$. We say that such $j$'s are {\em $b$-marked in $(X,\cS)$}.

First, the set of all feasible words $q \in \oo^d$ for a codebook $X \in \oo^{n \times d}$ and a malicious subset $\cS \subseteq [n]$ is denoted as

$$F(X,\cS) \eqdef \set{ q \in \oo^d \: | \: \forall j \in [d],\exists i \in \cS \: \colon \: q^j=x_i^j}.$$

Now, we are ready to formally define fingerprinting codes, similarly to \cite{BUV14}.

\begin{definition}[Fingerprinting Codes]\label{def:FPC}
For every $n, d \in \bbN$ and $\beta \in [0,1]$, a pair of algorithms $(\Gen,\Trace)$ is an {\rm $(n,d)$-fingerprinting code with security $\beta$} if $\Gen$ outputs a codebook $\bX \in \oo^{n \times d}$, and for every subset $\cS \subseteq [n]$ and every (possibly randomized) adversary $\Ac \colon \oo^{|\cS| \times d} \rightarrow \oo^d$ such that $\bq \sim \Ac(\bX_{\cS})$, it holds that    
\begin{itemize}
    \item $\pr{\bq \in F(\bX,\cS) \wedge \Trace(\bX,\bq)=\perp} \leq \beta$, and
    \item $\pr{\Trace(\bX,\bq) \in [n]\setminus \cS} \leq \beta$, 
\end{itemize}
where the probability is over the randomness of $\Gen$, $\Trace$ and $\Ac$. 
Moreover, the algorithm $\Gen$ can share an additional output $z \in \oo^*$ with the algorithm $\Trace$.
\end{definition}

\paragraph{Application: A Simple Fingerprinting Code.}
The following algorithm $\Gen$ generates a fingerprinting code of length $d$ as in \cref{lemma:FPL_high_dim} for $n$ users.%The code is generated column by column, that is, at the beginning 

\begin{algorithm}[$\Gen$]\label{FPC:gen}
	\item Input: Number of users $n \in \bbN$ and a confidence parameter $\beta \in [0,1]$.
	\item Operation:~
	\begin{enumerate}
		\item Let $d \geq \Theta(n^2\log^2 (n) \log(1/\beta))$ be the length of the code.
		
		\item Sample $(x_1,\ldots,x_n,z) \sim \cD'(n,d)$ (\cref{def:D}).

		\item Output $X = (x_i^{j})_{i\in [n], j \in [d]}$ and share $z = (z^{1},\ldots,z^{d})$ with $\Trace$.
	\end{enumerate}
\end{algorithm}

The next tracing algorithm $\Trace$ receives a codebook of $n$ codes of length $d$ (namely, a fingerprinting code as generated by the algorithm $\Gen$), a word of length $d$, possibly generated by a subset of malicious users using their codes, and a shared state of length $d$ shared by the algorithm $\Gen$.
It outputs a user from the malicious subset (if there is such user) or report that there is no such user, with high probability.

\begin{algorithm}[$\Trace$]\label{FPC:trace}
	\item Input: A codebook $X \in \oo^{n \times d}$, a word $q \in \oo^d$, and a shared state $z \in \oo^d$.%\Enote{removed the previous first checking}
	\item Operation:~
	\begin{enumerate}
		\item For every $i \in [n]$: If $\ip{x_i,q} - \ip{z,q} > \frac{0.2 d}{n \cdot \ln (5n)}$, output $i$.
		\item Else (no such $i \in [n]$), output $\perp$.
	\end{enumerate} 
\end{algorithm}

Next, we prove that the pair of algorithms $(\Gen,\Trace)$ is a fingerprinting code.

\begin{claim}
For $n \in \bbN$, $\beta \in [0,1]$, and $d \geq \Theta(n^2\log^2 (n) \log(1/\beta))$, the pair $(\Gen, \Trace)$ %described in 
(\cref{FPC:gen,FPC:trace}, respectively) is an $(n,d)$-fingerprinting code with security $\beta$ according to \cref{def:FPC}.
\end{claim}

\begin{proof}
We start by denoting $\bX \in \oo^{n \times d}$ as the output of $\Gen$ and $\bz \in \oo^d$ as the state shared by $\Gen$ with $\Trace$. 
We want to show that for every subset $\cS \subseteq [n]$ and every adversary $\Ac \colon \oo^{|\cS| \times d} \rightarrow \oo^d$ such that $\bq \sim \Ac(\bX_{\cS})$ is generated by $\Ac$, 
$$\pr{\bq \in F(\bX,\cS) \wedge \Trace(\bX,\bq)=\perp} \leq \beta \:\:\text{ and }\:\:
   \pr{\Trace(\bX,\bq) \in [n]\setminus \cS} \leq \beta.$$

To prove the first item, define the random variable $\bq'$ which is equal to $\bq$ if $\bq \in F(\bX,\cS)$ and otherwise takes a value in $F(\bX,\cS)$ (e.g., the first lexicographic vector there).
Furthermore, consider an algorithm $\Ac' \colon \oo^{\size{\cS} \times d} \rightarrow \oo^d$ that given $X_{\cS}$, samples $q \sim \Ac(X_{\cS})$ and checks if $q \in F(X,\cS)$ (i.e., perfectly agrees with all the marked columns of $X_{\cS}$). If it does, it outputs $q$. Otherwise, it outputs the first vector in $F(X,\cS)$ (which again can be done only using $\bX_{\cS}$). Observe that by definition, $\Ac'$ is strongly-accurate (\cref{def:strong_acc}) and also $\bq'$ has the same distribution as $\Ac'(\bX_{\cS})$. 
Hence
\begin{align*}
    \lefteqn{\pr{\bq \in F(\bX,\cS) \wedge \Trace(\bX,\bq)=\perp}}\\
    &= \pr{\bq \in F(\bX,\cS) \wedge \Trace(\bX,\bq)=\perp \mid \bq = \bq'} \cdot \pr{\bq = \bq'} + 0 \cdot \pr{\bq \neq \bq'}\\
    &= \pr{\Trace(\bX,\bq')=\perp \mid \bq = \bq'} \cdot \pr{\bq = \bq'}\\
    &\leq \pr{\Trace(\bX,\bq')=\perp}
    = \pr{\Trace(\bX,\Ac'(\bX_{\cS}))=\perp}
    \leq \beta,
\end{align*}
where the last inequality holds by \cref{lemma:FPL_high_dim}.

To prove the second item, fix $i \in [n] \setminus \cS$. Let $\bp^1,\ldots,\bp^d$ be the expectations that were sampled in the process of sampling $\bX,\bz$ (\cref{def:D}), and in the following assume they were fixed to some values $(p^1,\ldots,p^d)$. 
Recall that $\bq$ is a function of $\bX_{\cS}$. Therefore, $\bq$ is independent of $\bx_i$.
The same calculation done in \cref{eq:event2:general} yields that
\begin{align*}
    \pr{\ip{\bx_i,\bq} - \ip{\bz,\bq} > \frac{0.2 d}{n \ln(5n)}} \leq \frac{\beta}{20n}.
\end{align*}
The proof of the second item is now concluded by the union bound over $j \in [n]\setminus \cS$.

\end{proof}

\begin{remark}
In a {\em robust fingerprinting code}, we use a relaxed assumption for the set $F(X,\cS)$, as in \cite{BUV14}, and only require that the \jth element of $q$ is equal to the \jth element of $x_i$ for some $i\in \cS$ {\em with high probability}. 
It is possible to prove that $(\Gen,\Trace)$ is also a robust fingerprinting code. 
We can easily achieve robustness if we require the same fraction of mistakes in both $b$-marked sets  $\cJ^{b}_{X_{\cS}}$, but %regular 
robustness can be achieved anyway since the number of $1$-marked columns and the number of $(-1)$-marked columns are almost the same, and there are many such columns (a simple calculation yields that $\Omega(1/\log n)$ fraction of the columns are marked).    
\end{remark}

\end{document}